\pdfoutput=1
\PassOptionsToPackage{names,dvipsnames}{xcolor}
\documentclass[sigplan,screen,authoryear,10pt,acmsmall]{acmart}
\setcopyright{none}
\acmPrice{}
\acmDOI{}
\copyrightyear{}
\acmVolume{V}
\acmNumber{N}
\acmArticle{1}
\acmMonth{1}
\acmYear{2021}

\usepackage{utf}
\usepackage{anyfontsize} 
\usepackage{amsmath}     
\usepackage{amssymb}     
\usepackage{amsbsy}      
\usepackage{semantic}    
\usepackage{stmaryrd}   
\usepackage{mathpartir}  
\usepackage{braket}      
\usepackage{mathtools}   
\usepackage{thmtools}    
\usepackage{hyperref}
\usepackage[capitalise]{cleveref}
\usepackage{xspace}
\usepackage{mathrsfs} 
\usepackage{proof}
\usepackage{rotating}
\usepackage{tikz-cd}
\usetikzlibrary{arrows}
\usepackage{changepage}
\usepackage{centernot}
\usepackage{comment}
\usepackage{listings}
\usepackage{centernot}
\usepackage{multirow}
\usepackage{xparse}
\usepackage{pgothic} 
\usepackage{tensor} 
\usepackage{bbold} 
\usepackage{subcaption}
\usepackage[all]{xy}

\definecolor{darkblue}{rgb}{0.0, 0.0, 0.55}
\colorlet{darkgreen}{green!45!black}
\definecolor{darkred}{rgb}{0.55, 0.0, 0.0}
\definecolor{darkviolet}{rgb}{0.58, 0.0, 0.83}
\definecolor{lightblue}{rgb}{0.68, 0.85, 0.9}
\usepackage{lstcoq}

\makeatletter
\newcommand{\xMapsto}[2][]{\ext@arrow 0599{\Mapstofill@}{#1}{#2}}
\def\Mapstofill@{\arrowfill@{\Mapstochar\Relbar}\Relbar\Rightarrow}
\makeatother

\newcommand{\mynote}[2]
    {{\color{red} \fbox{\bfseries\sffamily\scriptsize#1}
    {\small$\blacktriangleright$\textsf{\emph{#2}}$\blacktriangleleft$}}~}

\definecolor{propcolor}{HTML}{3F7D31}
\definecolor{mypurple}{HTML}{5B069D}

\newcommand{\km}[1]{\mynote{KM}{#1}}

\newtheoremstyle{theoremthin}
{\smallskipamount}
{\smallskipamount}
{\itshape}
{0pt}
{\scshape}
{.}
{ }
{}

\theoremstyle{theoremthin}

\newtheorem{theorem}{Theorem}
\newtheorem{lemma}[theorem]{Lemma}

\newtheorem{corollary}[theorem]{Corollary}
\newtheorem{definition}{Definition}

\theoremstyle{remark}

\makeatletter
\renewenvironment{proof}[1][\proofname]{\par
  \pushQED{\qed}%
  \normalfont
  \topsep0pt \partopsep0pt 
  \trivlist
\item[\hskip\labelsep
  \itshape
  #1\@addpunct{.}]\ignorespaces
}{%
  \popQED\endtrivlist\@endpefalse
  \addvspace{2pt plus 2pt} 
}
\makeatother

\mathlig{[]}{\square}
\mathlig{.1}{._{1}}
\mathlig{.2}{._{2}}

\usepackage{scalerel}
\newcommand{\axiom}{\texttt{ax}}

\newcommand{\Coq}{\ensuremath{\mathrm{Coq}}\xspace}
\newcommand{\CIC}{\ensuremath{\mathsf{CIC}}\xspace}

\newcommand{\Agda}[0]{Agda\xspace}

\newcommand{\ie}{\emph{i.e.,}\xspace}

\newcommand{\eg}{\emph{e.g.,}\xspace}

\renewcommand{\P}{\operatorname{\Pi}}

\renewcommand{\l}{\operatorname{\lambda}}
 
\newcommand{\conv}{\equiv}

\newcommand{\pre}{\sqsubseteq}
\newcommand{\gpre}{\sqsupseteq}
\newcommand{\equiprecise}{\sqsupseteq\sqsubseteq}

\newcommand{\obsEquiv}{\approx^{obs}}
\newcommand{\obsApprox}{\preccurlyeq^{obs}}

\newcommand{\nat}{\mathbb{N}}
\newcommand{\bool}{\mathbb{B}}
\newcommand{\unit}{\mathbb{1}}
\newcommand{\ind}{\mathtt{ind}}
\newcommand{\rec}{\mathtt{rec}}

\newcommand{\ascName}{\mathrel{::}}
\newcommand{\asc}[2]{#1\ascName#2}

\newcommand{\id}[1]{\mathtt{id}_{#1}}
\newcommand{\Id}[3]{\operatorname{\operatorname{\mathtt{Id}}_{#1} #2} #3}

\usepackage{float}
\floatstyle{boxed} 
\restylefloat{figure}

\makeatletter
\def\mathcolor#1#{\@mathcolor{#1}}
\def\@mathcolor#1#2#3{%
	\protect\leavevmode
	\begingroup
	\color#1{#2}#3%
	\endgroup
}
\makeatother

\definecolor{grey}{rgb}{0.5, 0.5, 0.5}

\makeatletter
\def\slashedarrowfill@#1#2#3#4#5{%
  $\m@th\thickmuskip0mu\medmuskip\thickmuskip\thinmuskip\thickmuskip
  \relax#5#1\mkern-7mu%
  \cleaders\hbox{$#5\mkern-2mu#2\mkern-2mu$}\hfill
  \mathclap{#3}\mathclap{#2}%
  \cleaders\hbox{$#5\mkern-2mu#2\mkern-2mu$}\hfill
  \mkern-7mu#4$%
}
\def\rightslashedarrowfill@{%
  \slashedarrowfill@\relbar\relbar\mapstochar\rightarrow}
\newcommand\xslashedrightarrow[2][]{%
  \ext@arrow 0055{\rightslashedarrowfill@}{#1}{#2}}
\makeatother

\newcommand{\precision}[2]{\tensor[_{#1}]{\sqsubseteq}{_{#2}}}

\renewcommand{\id}[0]{\mathrm{id}}

\newif\ifappendix\appendixfalse

\newif\ifleftovers\leftoversfalse

\makeatletter

\DeclareDocumentCommand \inferrule { s O {} m m }{%
	\IfBooleanTF{#1}%
	{%
		\mpr@inferstar[#2]{#3}{#4}%
	}{%
		\mpr@inferrule[#2]{#3}{#4}%
	}%
	\IfValueT{#2}%
	{%
		\my@name@inferrule{#2}%
	}%
}
\NewDocumentCommand \my@name@inferrule { m }{%
	\def\@currentlabelname{\textsc{#1}}%
}

\newcommand*{\ilabel}[1]{%
  \phantomsection
  \label{#1}
}
\makeatother

\DeclareFontFamily{U}{FdSymbolA}{}
\DeclareFontShape{U}{FdSymbolA}{m}{n}{
  <-> s*[.28] FdSymbolA-Regular
}{}
\DeclareSymbolFont{fdsymbol}{U}{FdSymbolA}{m}{n}
\DeclareMathSymbol{\smallcirc}{\mathord}{fdsymbol}{"60}

\makeatletter
\newcommand{\hollowcolon}{\mathpalette\hollow@colon\relax}
\newcommand{\hollow@colon}[2]{%
  \mspace{1mu}%
  \vbox{%
    \hbox{$\m@th#1\smallcirc$}
    \nointerlineskip
    \kern.45ex
    \hbox{$\m@th#1\smallcirc$}
    \kern-.06ex
  }%
  \mspace{1mu}%
}
\makeatother

\definecolor{shadecolor}{gray}{0.93}

\usepackage{booktabs} 
\usepackage{enumitem}
\usepackage{pifont}

\newcommand{\subs}[2]{[#1 / #2]}

\definecolor{defgreen}{rgb}{0,0.6,0}

\crefname{section}{\textsection\!}{\textsection\!}
\crefname{proposition}{Prop.}{Props.}

\newcommand\metaop[1]{{\color{RoyalBlue}{\mathrm{#1}}}}
\newcommand\cstface[1]{{\color{ForestGreen}{\texttt{#1}}}}

\newcommand{\MuTT}[0]{\ensuremath{\mathsf{MuTT}}\xspace}
\newcommand{\MLTT}[0]{\ensuremath{\mathsf{MLTT}}\xspace}
\newcommand{\ExcTT}[0]{\ensuremath{\mathsf{ExcTT}}\xspace} \newcommand{\sortsymbol}[0]{\mathbb{S}} \newcommand{\sorts}[0]{\ensuremath{\sortsymbol}\xspace} \newcommand{\Type}[0]{\ensuremath{\mathbb{ty}}}
\newcommand{\Prop}[0]{\ensuremath{\mathbb{P}}\xspace}
\newcommand{\SProp}[0]{\ensuremath{\mathbb{sP}}\xspace}
\newcommand\Ax[0]{\ensuremath{\mathbb{Ax}}}

\newcommand\acst[3]{#1(#2;#3)}

\newcommand{\maxLevel}[2]{\mathrm{max}(#1, #2)}
\newcommand{\Jnosigctx}[2]{#1\, ; #2 \vdash}
\newcommand{\ofsort}[2]{#1~\hollowcolon~#2}
\newcommand{\oftype}[3]{#1 :^{#3} #2}
\newcommand{\Jnosigty}[4]{#1\, ; #2 \vdash \ofsort{#3}{#4}}
\newcommand{\Jnosigtm}[5]{#1\, ; #2 \vdash \oftype{#3}{#4}{#5}}
\newcommand{\Jnosigconvty}[5]{#1\, ; #2 \vdash #3 \equiv #4 ~\hollowcolon~ #5}
\newcommand{\Jnosigconvtm}[6]{#1\, ; #2 \vdash \oftype{#3 \equiv #4}{#5}{#6}}
\newcommand{\Jnosigredty}[5]{#1\, ; #2 \vdash #3 \Rightarrow #4 ~\hollowcolon~ #5}
\newcommand{\Jnosigredtm}[6]{#1\, ; #2 \vdash \oftype{#3 \Rightarrow #4}{#5}{#6}}
\newcommand{\Jnosigpat}[6]{#1\, ; #2 \vdash_{#6}\, \oftype{#3}{#4}{#5}}
\newcommand{\Jnosigsub}[4]{#1\, ; #2 \vdash #3 : #4}
\newcommand{\Jsub}[3]{\Jnosigsub{\Sigma{}}{#1}{#2}{#3}}
\newcommand{\Jctx}[1]{\Jnosigctx{\Sigma}{#1}}
\newcommand{\Jty}[3]{\Jnosigty{\Sigma}{#1}{#2}{#3}}
\newcommand{\Jtm}[4]{\Jnosigtm{\Sigma}{#1}{#2}{#3}{#4}}
\newcommand{\Jconvty}[4]{\Jnosigconvty{\Sigma}{#1}{#2}{#3}{#4}}
\newcommand{\Jconvtm}[5]{\Jnosigconvtm{\Sigma}{#1}{#2}{#3}{#4}{#5}}
\newcommand{\Jconvsub}[4]{\Jnosigconvtm{\Sigma}{#1}{#2}{#3}{#4}{}}
\newcommand{\Jpat}[5]{\Jnosigpat{\Sigma}{#1}{#2}{#3}{#4}{#5}}
\newcommand{\Jredty}[4]{\Jnosigredty{\Sigma}{#1}{#2}{#3}{#4}}
\newcommand{\Jredtm}[5]{\Jnosigredtm{\Sigma}{#1}{#2}{#3}{#4}{#5}}

\newcommand{\emptySig}[0]{\cdot}

\def\rewcomma{;}
\newcommand\RewRules[1]{\ensuremath{\mathcal{R}_{#1}}}
\newcommand\redrel[0]{\leadsto}
\newcommand\rewrule[7]{\ensuremath{(#2 : #1 \hookrightarrow #3 ,  #4(#5\rewcomma #6) ,   #7)}}
\newcommand\rewrulelin[5]{\ensuremath{(#1 ,  #2(#3\rewcomma #4) ,   #5)}}
\newcommand\rewrulelinx[4]{\ensuremath{(#1 ,  #2(#3) ,   #4)}}

\newcommand\rewtyping[8]{\ensuremath{#1 \, ; #3 : #2 \hookrightarrow #4 \vdash #5(#6\rewcomma #7) \redrel #8}}
\newcommand\rewtypinglin[6]{\ensuremath{#1 \, ; #2 \vdash #3(#4\rewcomma #5) \redrel #6}}
\newcommand\lin[0]{\mathrm{lin}}

\newcommand\whnf[1]{\mathop{\metaop{whnf}}#1}
\newcommand\neutral[1]{\mathop{\metaop{neutral}} #1}
\newcommand\linear[1]{\mathop{\metaop{linearizable}} #1}

\newcommand{\determine}[1]{\metaop{det}(#1)}
\newcommand{\progress}[1]{\metaop{progress}(#1)}
\newcommand{\isolated}[1]{\metaop{isolated}(#1)}
\newcommand{\reactR}[3]{\metaop{react}(#1,#2,#3)}
\newcommand{\react}[2]{\reactR{#1}{#2}{\RewRules{}}}

\newcommand{\ctxvar}[3]{#1 :^{#3} #2}
\newcommand\emptyContext{\cdot}
\newcommand\emptySub{\mathop{!}}

\newcommand{\univ}[1]{\ensuremath{\square^{#1}}}
\newcommand{\univi}[2]{\ensuremath{\square^{#1}_{#2}}}
\newcommand{\El}[2]{#2}

\newcommand\LRarg[1]{{\color{RedOrange}{#1}}}
\newcommand\LRpair[2]{#1 \mathbin{\LRarg{\mid}} \LRarg{#2}}
\newcommand{\LRnosigctx}[3]{#1\, ; #2 \Vdash #3}
\newcommand{\LRnosigty}[4]{#1\, ; #2 \Vdash \ofsort{#3}{#4}}
\newcommand{\LRnosigtm}[6]{#1\, ; #2 \Vdash \oftype{#3}{\LRpair{#4}{#6}}{#5} }
\newcommand{\LRnosigconvty}[6]{#1\, ; #2 \Vdash \ofsort{\LRpair{#3}{#6} \equiv #4}{#5}}
\newcommand{\LRnosigconvtm}[7]{#1\, ; #2 \Vdash \oftype{#3 \equiv #4}{\LRpair{#5}{#7}}{#6} }

\newcommand{\LRnosigsub}[5]{#1\, ; #2 \Vdash #3 : \LRpair{#4}{#5}}
\newcommand{\LRsub}[4]{\LRnosigsub{\Sigma{}}{#1}{#2}{#3}{#4}}
\newcommand{\LRctx}[2]{\LRnosigctx{\Sigma}{#1}{#2}}
\newcommand{\LRty}[3]{\LRnosigty{\Sigma}{#1}{#2}{#3}}
\newcommand{\LRtm}[5]{\LRnosigtm{\Sigma}{#1}{#2}{#3}{#4}{#5}}
\newcommand{\LRconvty}[5]{\LRnosigconvty{\Sigma}{#1}{#2}{#3}{#4}{#5}}
\newcommand{\LRconvtm}[6]{\LRnosigconvtm{\Sigma}{#1}{#2}{#3}{#4}{#5}{#6}}

\newcommand\LRredK[4]{\Sigma{} ; #1 \Vdash_K \oftype{#2}{#3}{#4}}
\newcommand\Kprop[4]{\Sigma{} ; #1 \Vdash^{\mathrm{nf}}_K \oftype{#2}{#3}{#4}}
\newcommand\LRneutraltm[4]{\Sigma{} ; #1 \Vdash_{\mathrm{ne}} \oftype{#2}{#3}{#4}}

\newcommand\type[0]{\ensuremath{\mathrm{Type}}\xspace{}}
\newcommand\prop[0]{\ensuremath{\mathrm{Prop}}\xspace{}}
\newcommand\sprop[0]{\ensuremath{\mathrm{SProp}}\xspace{}}

\newcommand\exc[0]{\ensuremath{\mathbb{Exc}}\xspace{}}

\newcommand{\raiseC}[0]{\cstface{raise}\xspace}
\newcommand{\raiseFull}[1]{\cstface{raise}(#1)}
\newcommand{\catchB}[0]{\ensuremath{\cstface{catch}_\bool}}
\newcommand{\catchBool}[5]{\cstface{catch}_{\bool}(#1, #2, #3, #4\rewcomma #5)}
\newcommand{\boolE}[0]{\ensuremath{\bool_\exc}\xspace}
\newcommand{\true}[0]{\cstface{true}\xspace}
\newcommand{\false}[0]{\cstface{false}\xspace}
\newcommand{\boolexn}[0]{\cstface{exc}\ensuremath{_\bool}\xspace}
\newcommand{\boolProp}[0]{\ensuremath{\bool_\Prop}\xspace}
\newcommand{\deamon}[0]{\maltese}

\newcommand{\Unit}[0]{\cstface{Unit}}

\newcommand\PTerm[2]{\P{} ( #1).~#2}

\newcommand{\BoxInd}[0]{\ensuremath{\mathbb{B}\mathrm{ox}}\xspace}
\newcommand{\BoxFull}[2]{\BoxInd_{#1}\,#2}
\newcommand{\boxC}[0]{\cstface{box}}
\newcommand{\boxFull}[2]{\boxC_{#1}\,#2}
\newcommand{\letbox}[3]{\texttt{let}\>\boxC\,#1 = #2 \>\texttt{in}\> #3}
\newcommand\unbox[0]{\cstface{unbox}}

\newcommand\Csts{\ensuremath{\mathcal{C}}}
\newcommand\Active[1]{\mathcal{A}_{#1}}
\newcommand\Inert[1]{\mathcal{I}_{#1}}
\newcommand\dom[1]{\ensuremath{\metaop{dom}(#1)}}
\newcommand\params[1]{\ensuremath{\metaop{params}(#1)}}
\newcommand\cod[1]{\ensuremath{\metaop{cod}(#1)}}
\newcommand\domsort[1]{\ensuremath{\sortsymbol^{\metaop{dom}}(#1)}}
\newcommand\codsort[1]{\ensuremath{\sortsymbol^{\metaop{cod}}(#1)}}

\newcommand\extendMuTT{\rightarrowtail}

\newcommand\occrec[2]{\metaop{\rec}_{#1}(#2)}
\newcommand\occrecsub[2]{{\color{RoyalBlue}{\rho{}}}^\metaop{\rec}_{#1}(#2)}
\newcommand\metavar[1]{#1}
\newcommand\metavarrec[2]{?^{\metaop{\rec}}#1[#2]}
\newcommand\erasepat[1]{{\color{RoyalBlue}{\epsilon}}(#1)}

\newcommand\List[1]{\mathop{\cstface{List}} #1}
\newcommand\consName[0]{\cstface{cons}}
\newcommand\consList[3]{\consName(#1,#2,#3)}
\newcommand\nilName[0]{\cstface{nil}}
\newcommand\nilList[1]{\nilName~#1}
\newcommand\listrecName[0]{\cstface{listRec}}

\def\Nat{\ensuremath{\mathbb{N}}}
\def\Zero{\cstface{0}}
\def\Succ{\cstface{S}}

\newcommand\pairName[0]{\cstface{pair}}
\newcommand\fstName[0]{\cstface{fst}}
\newcommand\sndName[0]{\cstface{snd}}

\def\Id{\texttt{Id}\xspace}
\def\Refl{\cstface{refl}\xspace}
\def\J{\cstface{J}\xspace}

\def\Coq{\textsc{Coq}\xspace}
\def\Agda{\textsc{Agda}\xspace}

\begin{document}

\title{The Multiverse: Logical Modularity for Proof Assistants}

\author{Kenji Maillard}
\affiliation{%
  \institution{Gallinette Project-Team, Inria}
  \city{Nantes}
  \country{France}
}
\author{Nicolas Margulies}
\affiliation{%
  \institution{ENS Paris-Saclay \& Gallinette Project-Team, Inria}
  \city{Nantes}
  \country{France}
}
\author{Matthieu Sozeau}
\affiliation{%
  \institution{Gallinette Project-Team, Inria}
  \city{Nantes}
  \country{France}
}
\author{Nicolas Tabareau}
\affiliation{%
  \institution{Gallinette Project-Team, Inria}
  \city{Nantes}
  \country{France}
}
\author{\'Eric Tanter}
\affiliation{%
  \institution{PLEIAD Lab, Computer Science Department (DCC), University of Chile}
  \city{Santiago}
  \country{Chile}
}

\begin{abstract} 
Proof assistants play a dual role as programming languages and
 logical systems. As programming languages, proof assistants offer standard
 modularity mechanisms such as first-class functions, type polymorphism and modules. As logical
 systems, however, modularity is lacking, and understandably so: incompatible
 reasoning principles---such as univalence and uniqueness of identity
 proofs---can indirectly lead to logical inconsistency when used in a given
 development, even when they appear to be confined to different modules. The
 lack of logical modularity in proof assistants also hinders the adoption of
 richer programming constructs, such as effects. We propose the multiverse, a
 general type-theoretic approach to endow proof assistants with logical
 modularity. The multiverse consists of multiple universe hierarchies that
 statically describe the reasoning principles and effects available to define a
 term at a given type. We identify sufficient conditions for this structuring
 to modularly ensure that incompatible principles do not interfere, and to
 locally restrict the power of dependent elimination when necessary. 
 This extensible approach
 generalizes the ad-hoc treatment of the sort of propositions in the \Coq proof
 assistant. We illustrate the power of the multiverse by describing the
 inclusion of \Coq-style propositions, the strict propositions of Gilbert et
 al., the exceptional type theory of P{\'e}drot and Tabareau, and general
 axiomatic extensions of the logic.
\end{abstract}

\maketitle


\section{Introduction}
\label{sec:intro}

Modularity is key to scalable software development~ \cite{Parnas:72}. As the adoption of proof assistants to write certified programs~\cite{cpdt} increases, software engineering aspects become crucial. Proof assistants are peculiar in that respect due to their dual role for programming and proving.
Indeed, while proof assistants such as \Coq~\cite{Coq:manual} and \Agda~\cite{norell:afp2008} offer traditional mechanisms for modular programming, including functional abstraction and modules, they lack modularity at the logical level. 
To illustrate, consider two logical principles which are known to be incompatible: 
{\em univalence}, a principle coming from Homotopy Type Theory~\cite{hottbook}, which provides a rich, computationally-relevant content to equality, and {\em uniqueness of identity proofs} (UIP), which considers two proofs of the same equality as necessarily equal. From both principles, one can derive a contradiction. But this conflict can be quite pernicious, as inconsistency can arise from seemingly harmless consequences of these principles. For instance, using univalence one can prove that there exists an equality on the type of
booleans, such that transporting $\mathtt{false}$ along this equality (noted
$e \# \mathtt{true}$ below) gives $\mathtt{false}$:
$$
\exists \ e : \bool = \bool, e \# \mathtt{true} = \mathtt{false}.
$$ 
This property does not mention univalence explicitly (in its ``interface''), yet combining it with UIP yields to inconsistency. In practice, this means that developers must make {\em global} commitments to certain reasoning principles in order to be sure that the underlying logic of their development is consistent. For instance, an \Agda development that imports univalence with the pragma
\texttt{\{-\# OPTIONS ----cubical \#-\}}~\cite{vezzosiAl:icfp2019} must be devoted to univalence and cannot be mixed with other incompatible extensions. Likewise, concerns about whether the use of classical principles is accepted or not must be made globally.
Therefore providing logical modularity is a whole new challenge in itself, not addressed by well-known modular programming constructs. The Curry-Howard correspondence between both words has its limits, unfortunately.

Another major consequence of the lack of logical modularity in proof assistants is that the use of {\em effects} in the programming language is typically demonized. However, software developers are well acquainted with the use of effects such as mutable state, exceptions, or control operators, and a great deal of effort involves addressing the mismatch between the pure world of proof assistants and real-world programs. Intuitively, the problem is that effects break a number of standard reasoning principles; for instance, the commutativity of addition on natural numbers is easy to prove by induction, but this induction principle is no longer valid in its full generality with effects, due to the relevance of evaluation order.
Recently, \citet{pedrotTabareau:popl2020} proved that mixing general induction
principles, a.k.a. {\em dependent elimination}, substitution, and effects, leads to inconsistency.
This incompatibility has a longer history, of course. The addition of
effects to a logical system can be traced back to double-negation translations~\cite{Glivenko29}, 
although the modern standpoint can be attributed to~\citet{moggi:iac1991}, as used for instance
in F$\star$~\cite{swamy2013verifying}. However, \citet{Barthe:2002} show that defining a typed CPS
translation preserving dependent elimination is out of reach,
and similarly, \citet{Herbelin05} proves that the theory behind the \Coq proof
assistant is inconsistent with computational classical logic under the guise of
a \texttt{call/cc} operator.
%
In retrospect, this incompatibility is an illustration of a very ancient issue: mixing
computational classical logic with the axiom of choice, whose intuitionistic version is a consequence of dependent elimination, is a well-known source of foundational problems~\cite{Martin-Lof06}.

Is all hope lost for logical modularity in proof assistants? Is there a way to encapsulate the use of effects to well-defined parts of a development so that they do not globally break logical consistency?
We answer these questions affirmatively, building upon a couple of
approaches developed in specific settings. First, in order to address
the issue of combining univalence and UIP, \citet{hts} proposed the
notion of homotopy type system---later revisited as two-level type theory by \citet{altenkirch2016extending}---which introduces two universe hierarchies in order to distinguish between so-called univalent types and strict types.
Second,~\citet{pedrot19} also propose the use of different universe hierarchies
to support consistent reasoning about effectful programs written in the Exceptional Type
Theory~\cite{pedrotTabareau:esop2018}.
%
Third, one can understand the well-known $\type$/$\prop$ distinction in \Coq under the same light: 
$\prop$---which is a one-level hierarchy indeed---lives ``apart'' from
the $\type$ hierarchy, with a restricted elimination schema from
$\prop$ into $\type$, known as {\em singleton
  elimination}. 
This restriction ensures that $\prop$ is compatible with proof irrelevance---a property assumed by the extraction mechanism~\cite{letouzey:phd}
to ensure computability of erased code---because otherwise 
one could prove that the type of booleans in $\prop$ has two distinct inhabitants.

While all these theories share the same substrate---Martin-Löf Type Theory~\cite{Martin-Lof-1973}---they come with their own pecularities and metatheoretical justifications, either developed on paper
\cite{altenkirch2016extending,GratzerKNB20,pedrotTabareau:esop2018,pedrot19} or mechanized 
\cite{AbelOV18,SozeauBFTW20}, involving a great deal of human effort and repeated work. 
Here, we develop a generic framework for defining, studying and combining such theories. 
For example, the addition of inductive types to a specific sort can in many cases be performed in a uniform manner. We additionally build a generic logical relation model that can accommodate multiple sorts, each with different sets of logical and computational principles, generalizing prior work by \citet{AbelOV18}.  To achieve this, we abstract the introduction of type constructors in a given sort and their inhabitants, along with their associated computational principles.

\paragraph*{Contributions}
Building upon this analysis and generalizing the idea of using a separate universe hierarchy to isolate a given reasoning principle or effect, this work develops the notion of the {\em multiverse} as a principled type-theoretic approach to endow proof assistants with logical modularity. 
The multiverse is a system with multiple universe hierarchies
that statically describe which principles and effects are available to
define a term at a given type. 
The multiverse permits 
the controlled use of incompatible reasoning principles in a development, where such principles can be used separately to establish different results about the same object of study, without any risk of unintended interference. Likewise, the multiverse makes it possible to extend the programming language of a proof assistant with effects, by locally restricting the power of dependent elimination in accordance with the considered effects.
Specifically:
\begin{itemize}
\item We introduce \MuTT{}, a dependent type theory parametrized by a
  description of multiple universes hierarchies and computational principles in \cref{sec:mutt}.
\item We illustrate the expressivity of the framework in \cref{sec:examples} with a presentation of
  inductive types, concrete instances providing \Coq-style propositions, the Exceptional
  Type Theory as well as general axiomatic extensions of the logic,
  and identify sufficient conditions for a universe to admit dependent elimination.
\item We show in \cref{sec:combination} that \MuTT{} indeed provides a modular
  framework: two independent parametrizations of \MuTT{} can be combined without
  endangering the metatheoretical properties of its core.
\item We prove important metatheoretical results on
  \MuTT{} that ensure consistency, canonicity and decidability of typechecking
  for any valid parametrization of the theory, showing that \MuTT is suitable as an idealized
  theory for proof assistants implementations (\cref{sec:metatheory}).
\item We briefly explain how the addition of extensionality
  principles fit in our framework and use it to describe an instance of \MuTT{}
  with strict propositions ($\sprop$)~\cite{GilbertCST19} as can be found in
  \Coq and \Agda, hence subsuming the theory of existing proof assistants.
\end{itemize}

Finally, \Cref{sec:related-work} discusses related work and \Cref{sec:conclusion} concludes.

\section{MuTT: Multiverse Type Theory}
\label{sec:mutt}

After a brief introduction to type theory, we present the syntax and typing of the Multiverse Type Theory (\MuTT{}), highlighting its parametrization, along with the expected conditions that a specific \MuTT{} parametrization must satisfy in order to be valid. Valid parametrizations of \MuTT{} yield a type theory suitable to serve as the basis for proof assistants.

\paragraph{Background} 
Martin-Löf Type Theory (\MLTT{})~\cite{itt} is a
dependent type theory featuring  dependent products (functions), dependent sums
(pairs) and identity types (equality). In \MLTT{} there is a single {\em sort} for all types, which is left implicit. The sort is represented by a {\em universe} constructor $\univ{}$ that classifies all types, including itself (represented as $\verb|Set|$ in \Agda). For example, a dependent function has a product type, and that product type itself has the type $\univ{}$.
More precisely, because the sort is {\em predicative}, it is structured as a stratified hierarchy of universes $\univ{}_i$, each at universe level $i$, so that $\univi{}{i}$ has type $\univi{}{i+1}$.

The Calculus of Inductive Constructions (\CIC{})~\cite{paulinMohring:appa2015}
generalizes \MLTT{} to include a schema for arbitrary inductive types and their
elimination principles.
For instance, the natural numbers can be defined in \CIC{} and one
can use the natural induction principle to reason about them. \CIC{}, like the
Calculus of Constructions~\cite{coquandHuet:ic1988}, features an additional sort
for {\em propositions}. Terms of a type of the proposition sort have a special
status as computationally-irrelevant information that can be erased through
extraction~\cite{letouzey:phd}. Additionally, the sort of propositions is
impredicative, in contrast to the sort of types for computationally-relevant
terms, and therefore the sort is not structured as a hierarchy. In \Coq, these two sorts are called $\type$ and $\prop$, respectively.

\subsection{Syntax and Typing}

\paragraph{\MuTT{} and parametrization}

Multiverse Type Theory (\MuTT{}) is an extensible variant of \MLTT{} with
multiple sorts. At its core, \MuTT features dependent functions and universes,
together with an extensible framework to define multiple sorts, and their
inhabitants. This means that the formal presentation of \MuTT is deeply
parametrized by a pair $\mathcal{P} = (\sorts, \Sigma)$:
\begin{itemize}
\item \MuTT is parametrized by a set \sorts of sorts, with a distinguished sort
  $\Type$ (read ``type''). The sort $\Type$ is primordial: it is necessarily
  present, and serves as the recipient to all universes, whatever their sort. A
  parametrization of \MuTT can include additional sorts. To any sort $s \in
  \sorts$ corresponds a hierarchy $\univi{s}{i}$ of universes, where $s$ is the
  sort of the universe and $i$ its level ($i \in
  \mathbb{N}$). Hierarchies in \MuTT are
  always predicative.
Any universe $\univi{s}{i}$ has sort $\Type$ at level $i+1$, or equivalently, has type 
$\univi{\Type}{i+1}$.
Additionally, the set $\sorts \setminus \Type$ comes with a predicate
$\isolated{s}$ which characterizes sorts whose information cannot be
used in $\Type$. 
\item To populate the sorts in \sorts, \MuTT is parametrized by two 
  sets underlying the signature $\Sigma$: a set of constants $\Csts$ 
  and a set of rewrite rules \RewRules{},
  which specify the computational aspect of the constants.
  All judgments of \MuTT{} are relative to the well-formed signature $\Sigma$
  that guarantees the well-formedness of types,
  constructors and eliminators, as well as determinism of the
  reductions and their completeness when the sorts involved are
  not isolated.
\end{itemize}

\paragraph{Syntax and notations}
The syntax of \MuTT, which is mostly standard except for the sort annotation on binders and universes, and the constants $c$ and $d$ from $\Csts$ (explained later on):
\[
    \begin{array}{llcl}
      \text{Terms} & t,u,p,A, B &::=& x \mid \lambda(\ctxvar{x}{A}{s}).t \mid t~u \mid c(\overline{t})
                \mid \acst{d}{\overline{t}}{u}\mid \P{}(\ctxvar{x}{A}{s})\,B  \mid \univ{s}\\
      \text{Substitutions}&\sigma, \overline{t} &::=& \emptySub \mid (\sigma, t) \mid t_1, \ldots, t_n
    \end{array}
  \] 
The empty substitution is noted $\emptySub$, and $(\sigma, t)$ is the extension
of a substitution $\sigma$ with a term $t$.
We use overlined variables $\overline{x}$ to denote a substitution as a sequence
of terms or variables $x_i$ and sometimes abuse context notations $\overline{x}
: \Gamma$ to make explicit the name of the variables bound in the context
$\Gamma$.
If $\Gamma$ is a context, $\Gamma_i$ is its component at the $i$th position and
$\Gamma_{<i}$ is its prefix excluding $\Gamma_i$. 
%
%
As usual, we write $A \rightarrow B$ for the
non-dependent version of the dependent product. 
We use the isomorphism between typing contexts and telescopes implicitly, 
i.e. if $Γ ⊢ t : Π~Δ, A$ then we can talk about $\overline{u} : Δ$ a 
well-typed instance/substitution for the context/telescope $Δ$.

\paragraph{Judgments}
\Cref{fig:MuTT-judgments} collects the defining judgments of \MuTT{}, which are all parametrized by a well-formed signature $\Sigma$.
To account for different sorts, traditional judgments have to be augmented with information about the sort.
For instance, the formation rule of dependent product should mention both the sorts and the
universe levels (Rule~\nameref{infrule:explicit-Pi}):
%
 \begin{mathpar}
   \inferrule[$\Pi$-explicit-level]
   {\Gamma \vdash A :^\Type \univi{s_1}{i} \\
     \Gamma,\ctxvar{x}{A}{s_1,i} \vdash B :^\Type \univi{s_2}{j} }
   {\Gamma \vdash
     {\PTerm{\ctxvar{x}{A}{s_1,i}}{B}}  :^\Type \univi{s_2}{\maxLevel{i}{j}}}
   \ilabel{infrule:explicit-Pi}
   \and
   \inferrule[$\Pi$-implicit-level]
   {\Gamma \vdash \ofsort{A}{s_1} \\
     \Gamma,\ctxvar{x}{A}{s_1} \vdash \ofsort{B}{s_2} }
   {\Gamma \vdash  \ofsort{\PTerm{\ctxvar{x}{A}{s_1}}{B}}{s_2}}
   \ilabel{infrule:implicit-Pi}
 \end{mathpar}
%
Because the formal treatment of universe levels is an orthogonal
concern that would obscure the presentation of the multisorted
extension of type theory, we adopt typical 
ambiguity~\cite{whitehead.russell:principia} and do not explicitly bind
universe levels $i$ in the rest of the paper, unless it helps understand the constructions.
Hence, 
we simply use a presentation as in Rule~\nameref{infrule:implicit-Pi}. This presentation 
makes more salient the main
information regarding the sorted version of $\Pi$: {\em the sort of
a $\Pi$ is the sort of its codomain}.
Therefore the typing judgment is written $\vdash t :^{s} A$. Likewise, the usual judgment $\vdash A$ of \MLTT{} which says that $A$ is a type is annotated with its sort $\vdash \ofsort{A}{s}$. Other judgments are decorated similarly.


%

%
\begin{figure}
\begin{small}
  \begin{tabular}{ll}
    $\Jctx{\Gamma}$& $\Gamma$ is a well formed context with respect to signature $\Sigma$\\
    $\Jsub{\Gamma}{\sigma}{\Delta}$ & $\sigma$ is a well formed substitution from $\Gamma$ to $\Delta$ \\
    $\Jty{\Gamma}{A}{s}$& $A$ is a well formed type at sort $s \in \sorts$ in context
                            $\Gamma$\\
    $\Jtm{\Gamma}{t}{A}{s}$& $t$ is a well formed term of type $A$ in context $\Gamma$\\
    $\Jconvty{\Gamma}{A}{B}{s}$& $A$ and $B$ are convertible types at sort $s$ in context $\Gamma$ \\
    $\Jconvtm{\Gamma}{t}{u}{A}{s}$& $t$ and $u$ are convertible at
                                    type $A$ and sort $s$ in context
                                    $\Gamma$ \\
    $\Jpat{\Gamma}{t}{A}{s}{d}$& $t$ is a well formed pattern of type $A$ in context $\Gamma$ for destructor  $d$\\
  \end{tabular}
  \end{small}
  \caption{Judgments of \MuTT{}}
  \label{fig:MuTT-judgments}
\end{figure}
\begin{figure}
\begin{small}
  \begin{mathpar}
    \inferrule[Ctx-Nil]{}{\Jctx{\cdot}}\ilabel{infrule:ctx-nil}
    \and
    \inferrule[Ctx-Cons]{\Jty{\Gamma}{A}{s}}{\Jctx{\Gamma, \ctxvar{x}{A}{s}}} \ilabel{infrule:ctx-cons}
    \and
    \inferrule[Empty-Sub]
    {}
    {\Jsub{\Gamma{}}{\emptySub}{\emptyContext}}
    \and
    \inferrule[Cons-Sub]
    { \Jsub{\Gamma} {\sigma}{\Delta} \\ \Jtm{\Gamma}{t}{A[\sigma]}{s} } 
    {\Jsub{\Gamma} { (\sigma , t) } { (\Delta, \ctxvar{x}{A}{s}) }}
    \ilabel{infrule:cons-sub}
    \\
    \inferrule[Var]{\Jty{\Gamma}{A}{s}}{\Jtm{\Gamma, \ctxvar{x}{A}{s}}{x}{A}{s}}\ilabel{infrule:var}
    \and
    \inferrule[Weak]{\Jtm{\Gamma}{t}{A}{s_1}\\\Jty{\Gamma}{B}{s_2}}{\Jtm{\Gamma, \ctxvar{y}{B}{s_2}}{t}{A}{s_1}}\ilabel{infrule:weak}
    \and
    \inferrule[Conv]{\Jtm{\Gamma}{t}{A}{s}\\\Jconvty{\Gamma{}}{A}{B}{s}}{\Jtm{\Gamma{}}{t}{B}{s}}\ilabel{infrule:conv}
    \\
    \inferrule[$\Pi$-Wf]{\Jty{\Gamma{}}{A}{s_1}\\\Jty{\Gamma{}, \ctxvar{x}{A}{s_1}}{B}{s_2}}{\Jty{\Gamma{}}{\PTerm{\ctxvar{x}{A}{s_1}}{B}}{s_2}}\ilabel{infrule:pi-wf}
    \\
    \inferrule[$\Pi$-Intro]{\Jtm{\Gamma{},
        \ctxvar{x}{A}{s_1}}{t}{B}{s_2}}{\Jtm{\Gamma{}}{\l(\ctxvar{x}{A}{s_1}).~t}{\PTerm{\ctxvar{x}{A}{s_1}}{B}}{s_2}}\ilabel{infrule:pi-intro}
    \and
    \inferrule[$\Pi$-Elim]{\Jtm{\Gamma{}}{t}{\PTerm{\ctxvar{x}{A}{s_1}}{B}}{s_2}\\\Jtm{\Gamma{}}{u}{A}{s_1}}
    {\Jtm{\Gamma{}}{t~u}{B\subs{u}{x}}{s_2}} \ilabel{infrule:pi-elim}
    \\
    \inferrule[Univ-Wf]{s \in \sorts\\\Jctx{\Gamma{}}}{\Jty{\Gamma{}}{\univ{s}}{\Type}}\ilabel{infrule:univ-wf}
    \and
    \inferrule[Univ-El]{\Jtm{\Gamma{}}{A}{\univ{s}}{\Type}}{\Jty{\Gamma{}}{\El{s}{A}}{s}} \ilabel{infrule:univ-el}
    \and
    \inferrule[El-Univ]{\Jty{\Gamma{}}{\El{s}{A}}{s}}{\Jtm{\Gamma{}}{A}{\univ{s}}{\Type}}\ilabel{infrule:el-univ}
    \\
    \inferrule[Inert-Type]{
      K \in \Sigma\\ 
      \cod{K} = \univ{s}\\
      \Jsub{\Gamma{}}{\overline{t}}{\params{K}}
    }{\Jty{\Gamma}{K(\overline{t})}{s}} \ilabel{infrule:inert-type}
    \\
    \inferrule[Inert-Term]{
      c \in \Sigma\\
      \cod{c} = K(\overline{u}) \\
      \Jsub{\Gamma{}}{\overline{p}}{\params{c}}\\
      \Jsub{\Gamma{}}{\overline{t}}{\dom{c}[\overline{p}]}\\
    }{
      \Jtm{\Gamma}{c(\overline{p},\overline{t})}{K(\overline{u})[\overline{p}]}{s}}\ilabel{infrule:inert-term}
    \and
    \inferrule[Active-Term]{
      d \in \Sigma\\
      \Jsub{\Gamma{}}{\overline{t} }{   \params{d}}\\
       \Jtm{\Gamma{}}{u}{\dom{d} [\overline{t}]}{\domsort{d}}
    }{\Jtm{\Gamma}{\acst{d}{\overline{t}}{u}}{\cod{d}[\overline{t},u]}{\codsort{d}}}\ilabel{infrule:active-term}
  \end{mathpar}
  \end{small}
  \caption{\MuTT{} typing rules (universe levels omitted)}
  \label{fig:MuTT-typing}
\end{figure}

\paragraph{Typing}
\Cref{fig:MuTT-typing} adapts the standard rules of \MLTT{} to account for
multiple sorts and universe hierarchies.
The rules for well-formed type environment, substitution, variables, 
dependent product, conversion, and universe are all standard from \MLTT, but extended with the sort information. For instance, Rule \nameref{infrule:pi-intro} specifies that a lambda term that takes an argument of type $A$ in sort $s_1$ can be typed in another sort $s_2$ provided its body is. Likewise, an application can happen in any sort $s_2$, even if the argument is typed in sort $s_1$.
%
Rule \nameref{infrule:univ-wf} states that all
universes are of sort $\Type$.
The other uniform choice would assign the sort $s$ to $\univ{s}$, but this rule
is not valid in general, for instance the
sort \Prop of proof-irrelevant types of \Coq has actually sort $\Type$ as
proof-irrelevant types themselves are not proof-irrelevant and cannot
be assigned sort \Prop itself. 
Following Coquand's reformulation of Russell's style presentation of
universes~\cite{Coquand19,Sterling19}, Rules \nameref{infrule:univ-el} and
\nameref{infrule:el-univ} together state that a type at sort $s$ can
equivalently be seen as a term of type $\univ{s}$.

The last three rules of \Cref{fig:MuTT-typing} deal with the parametrization of \MuTT with a set of constants $\Csts$, as mentioned above. 
%
%
Constants are further classified depending on whether they are {\em inert} or {\em active}, 
\ie~$\Csts = \Inert{} \cup \Active{}$. An inert constant $c \in \Inert{}$ does not trigger computation 
(\eg~a type, such as $\List{}$, or a constructor, such as $\consName$), while an active constant $d \in \Active{}$ must come with suitable rewrite rules that specify its computational content (\eg~the elimination principle of an inductive type, such as $\listrecName$). We now examine each in turn, using lists as a concrete parametrization example.

\paragraph{Inert constants}
Inert constants are used to introduce new types (Rule
\nameref{infrule:inert-type}), commonly noted $K$,
as well as new constructors for types (Rule \nameref{infrule:inert-term}).
An inert constant $c$ is described by contexts $\params{c}$ and $\dom{c}$ as
well as a type $\cod{c}$ specifying the parameters, the domain of recursive
occurences and codomain of the inert constant.
The distinction between inert types and constructors is done
relatively to the codomain \cod{c}, for which there are only two possibilities:
when $\cod{c} = \univ{s}$, then $c$ is an inert type of sort $s \in \sorts$,
and when $\cod{c} = K(\overline{u})$, then $c$ is a constructor of an inert type $K$ .
%
An applied inert type $K(\overline{t})$ is well-typed when $K$
appears in the signature $\Sigma$ and its arguments are of type
$\params{K}$.
Note that by the condition of well-formedness of $\Sigma$, we know
that $\dom{K} = \cdot$.
%
An applied inert term $c(\overline{p},\overline{t})$ of an inert type
$K$ has type $K(\overline{u})[\overline{p}]$ when $c$ appears in $\Sigma$, its arguments $\overline{p}$ are of type $\params{c}$,
and the recursive occurrences $\overline{t}$ are of
type $\dom{c}$. 
Note that here the substitution $\overline{u}$ providing the arguments
of $K$ is used to mediate between $\params{c}$ and $\params{K}$, which
may be different, in particular when encoding inductive types with
indices.
%
%
\begin{example}
  Consider the presentation of lists
  as an inert type of sort $\Type$. The type $\List{}$  is
  simply given by $\cod{\List{}} = \univ{\Type}$ and
  $\params{\List{}} = \ctxvar{A}{\univ{\Type}}{\Type}$ making any $\List{A}$ with
  $\ofsort{A}{\Type}$ well typed.
  The constructor $\consName$ is given by $\cod{\consName} =
  \List{A}$, $\params{\consName} =  \ctxvar{A}{\univ{\Type}}{\Type} ,
  \ctxvar{a}{A}{\Type}$ and $\dom{\consName} = \ctxvar{l}{\List{A}}{\Type}$. In that case, the substitution
  $\overline{u}$ from $\params{\consName}$ to $\params{\List{}}$ is 
  the first projection.
  Using Rule \nameref{infrule:inert-term}, we get that
  $\consList{A}{a}{l}$ is well-typed as an inhabitant of $\List{A}$ as
  soon as $A$ is a type in $\Type$, $a$ is an inhabitant of $A$ and
  $l$ is a list of $A$ itself. 
  Similarly for $\nilName$, we define $\cod{\nilName} = \List{A}$ and
  $\params{\nilName} = \ctxvar{A}{\univ{\Type}}{\Type}$, leaving
  $\dom{\nilName}$ empty.
\end{example}

\paragraph{Active constants}
An active constant is
described by four parameters \params{d}, \dom{d}, \cod{d} and
\codsort{d}. 
%
$\params{d}$ corresponds to the (possibly empty) context of parameters of $d$.
The domain $\dom{d}$ specifies the type of the scrutinee of $d$.
As for inert terms, we restrict the system to accept only two alternatives for $\dom{d}$: 
either $\dom{d} = \univ{s}$, in which case we define $\domsort{d} = \Type$; or
$\dom{d} = K (\overline{t'})$ with $K \in \Sigma$ and $\cod{K} =
\univ{s}$, in which case we define $\domsort{d} = s$.
An active term $\acst{d}{\overline{t}}{u}$ is well typed (Rule
\nameref{infrule:active-term}) when $d$ appears in $\Sigma$, 
$\overline{t}$ is of type $\params{d}$, the scrutinee $u$ is of type
$\dom{d}[\overline{t}]$. When this is the case, its return type is $\cod{d}[\overline{t},u]$.
%
%

%

%
\begin{example}
Coming back to the list example, the eliminator $\listrecName$ is presented as a term
  destructor with $\dom{\listrecName} = \List{A}$,
  $\params{\listrecName} = \ctxvar{A}{\univ{\Type}}{\Type}, \ctxvar{P}{A \rightarrow
    \univ{\Type}}{\Type}\ , \ctxvar{p_\nilName}{P~(\nilList{A})}{\Type} ,
  \ctxvar{p_\consName}{\PTerm{\ctxvar{a}{A}{\Type} ,
      \ctxvar{l}{\List{A}}{\Type}}{P~l \rightarrow
      P~(\consList{A}{a}{l})}}{\Type}$ 
  and with $\codsort{\listrecName} = \Type$ and $\cod{\listrecName} = P~u$, where $u$ is the variable
  associated to the scrutinee $\ctxvar{u}{\List{A}}{\Type}$ in 
  Rule~\nameref{infrule:active-term}.
\end{example}

\subsection{Conversion, Rewrite Rules, Reduction}

\begin{figure}
\begin{small}
  \begin{mathpar}
    \inferrule[Refl]{\Jtm{\Gamma}{t}{A}{s}}
    {\Jconvtm{\Gamma}{t}{t}{A}{s}} \ilabel{infrule:refl}
    \and
    \inferrule[Sym]{\Jconvtm{\Gamma}{t}{u}{A}{s}}
    {\Jconvtm{\Gamma}{u}{t}{A}{s}} \ilabel{infrule:sym}
    \and
    \inferrule[Trans]
    {\Jconvtm{\Gamma}{t}{u}{A}{s} \\ \Jconvtm{\Gamma}{u}{v}{A}{s}}
    {\Jconvtm{\Gamma}{t}{v}{A}{s}} \ilabel{infrule:trans}
    \and
    \inferrule[Conv-Univ-El]{\Jconvtm{\Gamma{}}{A}{B}{\univ{s}}{\Type}}{\Jconvty{\Gamma{}}{\El{s}{A}}{B}{s}} \ilabel{infrule:conv-univ-el}
    \and
    \inferrule[Conv-El-Univ]{\Jconvty{\Gamma{}}{\El{s}{A}}{B}{s}}{\Jconvtm{\Gamma{}}{A}{B}{\univ{s}}{\Type}}\ilabel{infrule:conv-el-univ}
    \and
    \inferrule[$\eta$-Conv]{\Jtm{\Gamma}{t,u}{\PTerm{\ctxvar{x}{A}{s_1}}{B}}{s_2}
      \\\Jconvtm{\Gamma, \ctxvar{x}{A}{s_1}}{t \ x}{u \ x}{B}{s_2}}
    {\Jconvtm{\Gamma}{t}{u}{\PTerm{\ctxvar{x}{A}{s_1}}{B}}{s_2}} \ilabel{infrule:eta-conv}
    \and
    \inferrule[Red-Conv]{\Jredtm{\Gamma}{t}{u}{A}{s}}
    {\Jconvtm{\Gamma}{t}{u}{A}{s}} \ilabel{infrule:red-conv}
    \and
    \inferrule[App-Conv]{\Jconvtm{\Gamma}{t}{t'}{{\PTerm{\ctxvar{x}{A}{s_1}}{B}}}{s_2} \\ \Jconvtm{\Gamma}{u}{u'}{A}{s_1}}
    {\Jconvtm{\Gamma}{t~u}{t'~u'}{B\subs{u}{x}}{s_2}}
    \ilabel{infrule:app-conv}
    \hspace{4em} \mbox{(other congruence rules omitted)}
  \end{mathpar}
  \end{small}
  \caption{Conversion for \MuTT}
  \label{fig:MuTT-conv}
\end{figure}

\paragraph{Conversion}

\cref{fig:MuTT-conv} presents the conversion rules of \MuTT.
To insist on the central role of reduction in the theory, conversion
is defined as the transitive, reflexive, symmetric closure by
congruence of reduction.
Rule~\nameref{infrule:app-conv} illustrates the congruence rule for
application. 
Reduction is embedded inside conversion through Rule~\nameref{infrule:red-conv}.
Conversion additionally provides a way to navigate between type and
term conversion (Rules~\nameref{infrule:conv-univ-el} and
\nameref{infrule:conv-el-univ}) and also features $\eta$-conversion for
functions (Rule~\nameref{infrule:eta-conv}). 
The definition of reduction itself is parametrized
by {\em rewrite rules}~\cite{CockxTW21} that can be added to \MuTT.

\paragraph{Reduction}

\cref{fig:MuTT-red} describes the notion of reduction in \MuTT, which
features the usual notion of $\beta$-reduction
(Rule~\nameref{infrule:beta-conv}). Because reduction is itself typed,
one needs to add a compatibility rule with conversion at the level of
types (Rule~\nameref{infrule:conv-red}).
Congruence rules on applications and active terms are turned into
substitution rules in order to emulate reduction to {\em weak-head normal forms} (whnfs).
Indeed, the classification between inert $c \in \Inert{}$ and active constants
$d \in\Active{}$ gives rise to well-behaved
definitions of whnfs and {\em neutral forms}, two key notions to establish the
metatheory.
  \[
    \begin{array}{l@{\qquad}lclr}
      \whnf{t}& w &::=& ne \mid \P{}(\ctxvar{x}{A}{s})\,B \mid
      \l(\ctxvar{x}{A}{s}).\,t\mid \univ{s} \mid c(\overline{p},\overline{t}) & \\
      \neutral{t}& ne &::=& x \mid ne~t \mid \acst{d}{\overline{t}}{ne} \mid \acst{d}{\overline{t}}{c(\overline{p},\overline{u})} &
       \hspace{3em} \lnot \react{d}{c}
    \end{array}
  \]

These notions are at the heart of the logical relation given
in~\cref{sec:metatheory}.
For now, it is enough to know that weak-head normal forms correspond to terms
that can \emph{not} be head-reduced, of which neutral terms are the particular
cases where the term may not stay in weak-head normal form after substitution.
An inert type or inert term is always a whnf.
An active term is neutral when its scrutinee is in whnf and there is
no rewrite rule in \RewRules{} that can be fired, as expressed by the
following definition:
$$
\react{d}{c} \quad \eqbydef \quad 
  \exists
    \rewrule{\Delta}{\sigma}{\Delta_{\lin{}}}{d}{\overline{x}}{pat}{r} \in \RewRules{} , \quad
    pat = c~q
$$
So to achieve whnf reduction, we need to add a substitution rule
(Rule~\nameref{infrule:app-subs}) that
reduces the left-hand side of an application and the
scrutinee of an active term until it reaches a whnf.
Finally, each rewrite rule in \RewRules{} is turned into a reduction
rule using rule~\nameref{infrule:rew-red}, which basically considers
any correct (linear) instantiation of a rewrite rule.

\begin{figure}
\begin{small}
  \begin{mathpar}
        \inferrule[$\beta$-Red]{
      \Jtm{\Gamma,\ctxvar{x}{A}{s_1}}{t}{B}{s_2} \\ \Jtm{\Gamma}{u}{A}{s_1}
    }{\Jredtm{\Gamma}{(\l(\ctxvar{x}{A}{s_1}).~t)~u}{t\subs{u}{x}}{B\subs{u}{x}}{s_2}} \ilabel{infrule:beta-conv}
\and
    \inferrule[Conv-Red]{
      \Jredtm{\Gamma}{t}{u}{A}{s} \\ \Jconvty{\Gamma}{A}{B}{s}
    }{\Jredtm{\Gamma}{t}{u}{B}{s}} \ilabel{infrule:conv-red}
    \and
        \inferrule[App-Subs]{\Jredtm{\Gamma}{t}{t'}{{\PTerm{\ctxvar{x}{A}{s_1}}{B}}}{s_2} \\ \Jtm{\Gamma}{u}{A}{s_1}}
    {\Jredtm{\Gamma}{t~u}{t'~u}{B\subs{u}{x}}{s_2}}
    \ilabel{infrule:app-subs}
    \and
    \inferrule[Active-Subs]{ \Jsub{\Gamma}{\overline{p}}{\params{d}} \\
    \Jredtm{\Gamma}{t}{t'}{\dom{d}[\overline{p}]}{\domsort{d}}} 
  {\Jredtm{\Gamma}{\acst{d}{\overline{p}}{t}}{\acst{d}{\overline{p}}{t'}}
    {\cod{d}[\overline{p},a]}{\codsort{d}}}
    \ilabel{infrule:active-subs}
    \and    
    \inferrule[Rew-Red]{
      \rewrule{\Delta}{\sigma}{\Delta_{\lin{}}}{d}{\overline{x}}{pat}{r}
      \in \RewRules{}  \\
      \Jsub{\Gamma{}}{\sigma'{}}{\Delta_{\lin{}}}\\
      (\overline{t}, a) = (\overline{x},  \erasepat{pat}) [\sigma'{}]
      \\
      \Jtm{\Gamma}{d(\overline{t}, a)}{\cod{d}[\overline{t},
        a]}{\codsort{d}}
    }{\Jredtm{\Gamma}{\acst{d}{\overline{t}}{a}}{r[\sigma',
        \occrecsub{d}{pat}[\sigma']]}{\cod{d}[\overline{t},
        a]}{\codsort{d}}} \ilabel{infrule:rew-red}
  \end{mathpar}
  \end{small}
  \caption{Reduction Rules for \MuTT}
  \label{fig:MuTT-red}
\end{figure}

\begin{figure}
\begin{small}
  \flushleft
  
  \begin{align}
    \label{eq:pattern}
    p\enspace &::=\enspace c(q_1,\ldots, q_n) \mid \P{}^s~q_1~q_2 \tag{pattern}\\
    q\enspace &::=\enspace \metavar{x} \mid {\metavarrec{z}{\sigma}}
    \tag{meta-variable}
  \end{align}

  \fbox{$\Jpat{\Gamma}{p}{A}{s}{d}$} \emph{Pattern typing}.
  
  \begin{mathpar}
    \inferrule[Inert-Pat]{
      \Jpat{\Gamma}{\overline{q}}{\params{c}}{}{d}\\
      \Jpat{\Gamma}{\overline{q'}}{\dom{c}}{}{d}[\overline{\erasepat{q}}]
    }
    {\Jpat{\Gamma}{c(\overline{q},\overline{q'})}{\cod{c}[\overline{\erasepat{q}}]}{\codsort{c}}{d}}
    \ilabel{infrule:inert-pat}
    \and
    \inferrule[MetaVar-Pat]{
      \ctxvar{x}{A}{s} \in \Gamma}
    {\Jpat{\Gamma}{\metavar{x}}{A}{s}{d}}
    \ilabel{infrule:metavar-pat} \and
    \inferrule[Sub-Pat]
    { \Jpat{\Gamma} {\sigma}{\Delta}{}{d} \\ \Jpat{\Gamma}{t}{A[\sigma]}{s}{d}}
    {\Jpat{\Gamma} { (\sigma , t) } { (\Delta, \ctxvar{x}{A}{s}) }{}{d}}
    \ilabel{infrule:sub-pat}
    \and
    \inferrule[$\P^{s_1}$-Pat]{
      \Jpat{\Gamma}{q_1}{\univ{s_1}}{\Type}{d}\quad
      \Jpat{\Gamma}{q_2}{\erasepat{q_1} {\to} \univ{s_2}}{\Type}{d}}
    {\Jpat{\Gamma}{\P^{s_1}~q_1~q_2}{\univ{s_2}}{\Type}{d}}
    \ilabel{infrule:pi-pat}
    \and
    \inferrule[MetaVarRec-Pat-Term]{
      \Jsub{\Gamma{}}{\sigma{}}{\params{d}}
      \quad
      \ctxvar{z}{\dom{d}[\sigma]}{s} \in \Gamma}
    {\Jpat{\Gamma}{\metavarrec{z}{\sigma}}{\dom{d}[\sigma]}{s}{d}}
    \ilabel{infrule:metavarrec-pat-term} \\

  \end{mathpar}

  \fbox{$\erasepat{pat}, \occrec{d}{pat}, \occrecsub{d}{pat}$}
  \emph{erasure to terms, context and substitution of recursive hypothesis.}
  \[\arraycolsep=1pt
    \begin{array}{lcl}
      \erasepat{\metavar{x}} & := & x
      \\
      \erasepat{\metavarrec{z}{\sigma}} &:=& z 
      \\
      \erasepat{\P{}^s\,q_1\,\metavar{x}} &:=&
                                               \P{}(\ctxvar{y}{\erasepat{q_1}}{s})\,(x~y)
      \\
      \erasepat{\P{}^s\,q_1\,\metavarrec{z}{\sigma}}&:=&
                                                         \P{}(\ctxvar{y}{\erasepat{q_1}}{s})\,(z~y)
      \\
      \erasepat{c(q_1, \ldots, q_n)} &:=& c(\erasepat{q_1},\ldots, \erasepat{q_n})
    \end{array}
    \begin{array}{lcl}
      \occrec{d}{\metavar{x}} &:= & \emptyContext
      \\
      \occrec{d}{\metavarrec{z}{\sigma}} &:=&z^{\rec} : \cod{d}[\sigma,z]
      \\
      \occrec{d}{\P{}^s\,q_1\,\metavar{x}} &:=& \occrec{d}{q_1}
      \\
      \occrec{d}{\P{}^s\,q_1\,\metavarrec{z}{\sigma}} &:=& z^{\rec} :
                                                           \P{}(\ctxvar{y}{\erasepat{q_1}}{s})\,
                                                           \cod{d}[\sigma,z~y]
      \\
      \occrec{d}{c(q_1, \ldots, q_n)} &:=& \occrec{d}{q_1}, \ldots, \occrec{d}{q_n}
    \end{array}
  \]%
  \[
    \begin{array}{lcl}
      \occrecsub{d}{\metavar{x}} &:=& \emptySub \\
      \occrecsub{d}{\metavarrec{z}{\sigma}}&:=& \acst{d}{\sigma}{z}\\
      \occrecsub{d}{\P{}^{s}\,q_1\,\metavar{x}} &:=& \occrecsub{d}{q_1}\\
      \occrecsub{d}{\P{}^{s}\,q_1\,\metavarrec{z}{\sigma}} &:= &
        \left( \occrecsub{d}{q_1}, \l{}(y : \erasepat{q_1}).\,d(\sigma,
        z\,y)\right)\\
      \occrecsub{d}{c(q_1, \ldots, q_n)}&:=& \occrecsub{d}{q_1},\ldots, \occrecsub{d}{q_n}
    \end{array}
  \]
  \end{small}
  \caption[Patterns typing]{Syntax, typing rules and functions on
    patterns}
  \label{fig:MuTT-patterns}
\end{figure}

\paragraph{Patterns and Rewrite Rules}
Active constants are interesting when associated to rewrite rules, 
which make it possible to extend the conversion of \MuTT{}.
A rewrite rule
$$\rewrule{\Delta}{\sigma}{\Delta_{\lin{}}}{d}{\overline{x}}{pat}{r} \in \RewRules{}$$
is given by a left-hand side,
characterized by its head symbol
$d$, which must be an active constant, a renaming
$\overline{x}$, and a pattern $pat$ for its scrutinee, and a
right-hand side $r$, which is simply a term.  
It ensures that its left- and right-hand sides are convertible for
any instance of the (linear) context~$Δ_{\lin{}}$ for which the
left-hand side is well-typed (\nameref{infrule:rew-red} in \cref{fig:MuTT-red}).
However, in general, the use of a linear context is not enough to
guarantee that the right-hand side is well typed, because this may
rely on some auxiliary conversions ensured by typing.
On the other hand, the use of a non-linear context alone is not
sufficient either because the reduction rule
\nameref{infrule:rew-red} would require a prohibitive use of conversion in the definition of
reduction.
%
To remedy to this tension, a rewrite rule has two contexts $\Delta$ and $\Delta_{\lin{}}$ with
a renaming $\sigma$ (\ie a substitution with only variables).
The point of the context $\Delta$ and substitution $\sigma$ is to
provide a non-linear version of the rewrite rule, which can be used for
typing purposes, with a side condition that every well-typed linear
occurrence of the rule is actually well-typed as a non-linear
occurrence.
When the rewrite rule is \emph{linear}, that is when $\Delta = \Delta_{\lin{}}$
and $\sigma = id$,
we simply write $\rewrulelin{\Delta}{d}{\overline{x}}{pat}{r}$.
The syntax and typing rules of patterns are described in \cref{fig:MuTT-patterns}. A pattern
consists either of an inert constant $c \in \Csts$
 or a $\Pi$ applied to metavariables, while a metavariable can be either a variable $\metavar{x}$ or
a recursive occurrence $\metavarrec{z}{\sigma}$.
The purpose of recursive occurrences is twofold. When typing a
pattern, a recursive occurrence has the type of the domain head
symbol $\dom{d}$, when its substitution $\sigma$ is a correct
substitution for the parameters of $d$ (Rule
\nameref{infrule:metavarrec-pat-term}). It additionally enforces that
the recursive occurrences $\metavarrec{z}{\sigma}$ correspond to
a variable $z$ with the same type in the context, which may be used in
the right-hand side of a rewrite rule.
Apart from these two rules, the typing rules of patterns just mimic the
typing rules for terms and substitutions. 

From the pattern $pat$, we define three functions: a telescope
$\occrec{d}{pat}$ extending the context $\Delta$, collecting the recursive occurrences in
the pattern, which is used to typecheck the recursive calls to \(d\) in the
right-hand side of the rewrite rule; a substitution
$\Jsub{\Delta}{\occrecsub{d}{pat}}{\occrec{d}{pat}}$ instantiating all
occurrences of $\metavarrec{z}{\sigma}$ with the corresponding intended instance
of $d$; and an erasure function $\erasepat{pat}$ computing the underlying term
of a pattern.
\begin{figure}
\begin{small}
  \begin{mathpar}
    \inferrule[Valid-Rew]{
      \Jsub{\Delta}{\sigma}{\Delta_{\lin{}}} \\
      \Jsub{\Delta_{\lin{}}}{\overline{x}}{\params{d}} \\ 
      \Jpat{\Delta_{\lin{}}}{pat}{A}{\domsort{d}}{d}\\
      A[\sigma] = \dom{d}[\overline{x}][\sigma] \\ 
      \Jtm{\Delta,
        \occrec{d}{pat}[\sigma]}{r}{\cod{d}[\overline{x},\erasepat{pat}][\sigma]}{\codsort{d}} \\
      \linear{\rewrule{\Delta}{\sigma}{\Delta_{\lin{}}}{d}{\overline{x}}{pat}{r}} 
    }
    {\rewtyping{\Sigma}{\Delta}{\sigma}{\Delta_{\lin{}}}{d}{\overline{x}}{pat}{r}}
    \ilabel{infrule:rew-typing-term}
  \end{mathpar}
  \end{small}
  \caption{Typing for rewrite rules}
  \label{fig:typing-rewrite}
\end{figure}
Finally, Rule \nameref{infrule:rew-typing-term} specifies when a
rewrite rule is valid (\cref{fig:typing-rewrite}).
The rule checks that: (i) the renaming $\sigma$ is well-typed,
(ii) the rewrite rule is linearizable
(iii) the parameters $\overline{x}$ form a well-typed renaming to
$\params{d}$ in context
$\Delta_{\lin{}}$,
 (iv) the pattern $pat$ is a well-typed pattern in context
$\Delta_{\lin{}}$,
(v) and that the right-hand side $r$ is a well-typed term, in the
context $\Delta$ extended with the information that the recursive occurrences
appearing in $pat$ are now seen as variables living in the codomain
of $d$ (modulo the renaming $\sigma$).

In the rule, the context $\Delta_{\lin{}}$ is used to typecheck separately the
arguments $\overline{x}$ and $pat$ of the left-hand side linearly, ensuring that
weak-head reduction is enough to detect when a rewrite rule can fire.
However, we need to allow non-linearity (described by the renaming $\sigma$) in
order to enforce that the types in the left-hand side agree up to the renaming
 and to type-check the right hand-side, this even for the simple
 example of lists (see \cref{example:rew-list}).
To guarantee a posteriori that linear and non-linear matching are
equivalent, we introduce the notion of a linearizable rewrite rule.

\begin{definition}[Linearizable rewrite rule]
\label{def:linearizability}
A rewrite rule
$\rewrule{\Delta}{\sigma}{\Delta_{\lin{}}}{d}{\overline{x}}{pat}{r} $
is linearizable, noted 
$
\linear{\rewrule{\Delta}{\sigma}{\Delta_{\lin{}}}{d}{\overline{x}}{pat}{r}}
$, when
\begin{itemize}
\item[(a)] every variable in $\Delta_{\lin{}}$ occurs exactly once
in $(\overline{x},\erasepat{pat})$
and either
\item[(b)] $\dom{d} = \univ{s}$ and the
rewrite rule is linear; or
\item[(b')]  $\dom{d} = K(\overline{u_d})$ and the
  following holds where $pat = c(\overline{q})$ and $\cod{c} = K(\overline{u_c})$:
$$
    \Jconvsub{\Delta_{\lin{}}}{\overline{u_c}[\overline{q}]}{\overline{u_d}[\overline{x}]}{\params{K}}
    \Rightarrow \exists \tau , \, 
    \tau[\sigma] = \id_{\Delta}  \wedge \Jconvsub{\Delta_{\lin}}{\sigma[\tau]}{\id_{\Delta_{\lin}}}{\Delta_{\lin}}.
$$
\end{itemize}
\end{definition}
The condition $(b)$ says that every eliminator on a universe must
be linear. The condition $(b')$ says that using the conversion
constraints collected from the fact that the left-hand side
type-checks, one can show that actually the renaming $\sigma$ admits an inverse,
in other word, linear and non-linear matching coincides up to conversion.

The precise type for the recursive occurrences is computed by the
function $\occrec{d}{pat}$.
Thus, when typing a rewrite rule, a recursive occurrence
$\metavarrec{z}{\sigma}$ of a pattern is seen as variable $z$ of
type $\dom{d}[\sigma]$ when typing both sides, and additionally as a variable
$z^{\rec}$ of type $\cod{d}[\sigma,z]$, representing the result of applying
$d$ to the variable, when typing the right-hand side.
This allows us to encode recursive calls to $d$ that may occur on the
right-hand side of a rewrite rule, as illustrated by the following example.

\begin{example}
  \label{example:rew-list}
  Coming back to the representation of lists, the two rewrite rules
  for $\listrecName$ are:
  \begin{small}
  $$\RewRules{\List{}} \enspace \eqbydef \enspace
  \left |
    \begin{array}{l}
    \rewtyping{.}{\Delta}{\overline{p},A}{(\Delta,B:\univ{\Type})}{\listrecName}{\overline{p}}{\nilList{\metavar{B}}}{p_\nilName} \\
    \rewtyping{.}{\Delta_{\mathrm{cons}}}{\overline{p},A,a,l}
    {\Delta_{\lin{}}}{\listrecName}{\overline{p}}{\consList{\metavar{B}}{\metavar{a}}{(\metavarrec{l}{\overline{p}})}}
    {p_\consName~a~l~l^{\rec}} 
    \end{array}
    \right.
  $$
  with $\overline{p} = A, P, p_\nilName, p_\consName$ and 
  $$
     \Delta =
  \overline{p} :\params{\listrecName} \quad
  \Delta_{\mathrm{cons}} =
      \Delta,\ctxvar{a}{A}{\Type},\ctxvar{l}{\List{A}}{\Type} \quad
  \Delta_{\lin{}} =
      \Delta,B:\univ{\Type},\ctxvar{a}{B}{\Type},\ctxvar{l}{\List{B}}{\Type}.
    $$  
    \end{small}
  The first rewrite rule is valid because as the right-hand side
  $p_\nilName$ has type $P~(\nilList{A})$, and the linearizability
  condition amounts to show that $\overline{p},A$ has a retraction,
  knowing that $A \equiv B$, which is direct by mapping $B$ to $A$. 
  The second rule makes use of a recursive occurrence
  $\metavarrec{l}{\overline{p}}$, which corresponds both to the variable $l$ of
  type $\List{A}$ and, for the right-hand side, to the variable
  $l^{\rec}$ of type $P~l$ representing the recursive call on $l$.
  Therefore, the right-hand side is well-typed, with type
  $P~(\consList{A}{a}{l})$.
  Linearizability of the rule is similar to the case of $\nilList{}$.
\end{example}

%
\subsection{Well-formed Signature}
\label{sec:well-form-sign}
We now turn to the definition of a well-formed signature.
The signature imposes constraints on each inert and active constant that
structure their global behavior and interaction with the whole system.
Inert constants $K$ building types are then classified as \textbf{positive} if they come
with a set $\Inert{K}$ of inert constant called \textbf{constructors} to introduce them,
or \textbf{negative} if they come with a set $\Active{K}$ of active constants
called \textbf{observations}.
An active constant defined on a positive type is then called an
\textbf{eliminator}, whereas an inert constant inhabiting a negative type is
called a \textbf{builder}.
Universes are treated as positive types, with the exception that their
constructors are open-ended and consist of any type constant of the adequate sort.
%
A well-formed signature
is either the empty signature $\emptySig$, a well-formed signature $\Sigma$
extended with a well-formed positive type $(K, \Inert{K})$, negative type $(K,
\Active{K})$, an eliminator $(d, \RewRules{d})$ on a universe or positive type,
or a builder $(c, \RewRules{c})$ on a negative type.
We say that a constant or a rewrite rule $\kappa$ belongs to a well-formed signature
$\Sigma$, noted $\kappa \in \Sigma$, if it appears in one of its components.
Well-formedness for rewriting relies on three key properties---determinism of the
set of rewrite rules, progress and isolation---detailed next.

\paragraph{Deterministic rewrite rules}

The first property ensures that the notion of whnf reduction defined
in \cref{fig:MuTT-red} is deterministic, which is crucial to
easily get confluence of the system.\footnote{We could adopt a
  more permissive condition for confluence~\cite{CockxTW21}, but this is not central here.}

\begin{definition}[Deterministic rewrite rules]
  \label{def:rewrite-determinism}
  A set of rewrite rules $\RewRules{}$ is deterministic, noted $\determine{\RewRules{}}$, if any two
  rewrite rules from $\RewRules{}$
  with common head symbol and common head-constant in their
  patterns have the same right-hand sides:
  \begin{align*}
    \determine{\mathcal{R}} \enspace\eqbydef\enspace
    \forall \rewrule{\Delta}{\sigma}{\Delta_{\lin{}}}{d}{\overline{x}}{pat}{r} ,
    \rewrule{\Delta'}{\sigma'}{\Delta'_{\lin{}}}{d'}{\overline{x'}}{pat'}{r'} \in
    \RewRules{},\\
    d = d' \enspace \wedge \enspace
    pat = c~q \enspace\wedge\enspace pat' = c ~ q'
    \enspace\Rightarrow\enspace r = r'
  \end{align*}
\end{definition}


\paragraph*{Rewrite progress}

The second notion that helps us characterize well-formed signatures is
{\em progress} between a set of inert terms and a set of active terms with respect to a set of
rewrite rules.

\begin{definition}[Rewrite progress]
  \label{def:rewrite-progress}
  A set of inert constants $I$ and a set of active constants $A$ satisfy rewrite progress with respect to a set of rewrite rules
  $\RewRules{}$, noted $\progress{I, A, \RewRules{}}$,  
  if every active constants in $A$ reacts to every constant in $I$ according to $\RewRules{}$:
  \begin{align*}
    \progress{I,A,\RewRules{}}\quad \eqbydef \quad \forall c \in I, d \in A, \enspace \reactR{d}{c}{\RewRules{}}
  \end{align*}
  When $I = \{c\}$ or $A = \{d\}$ are singleton, we note
  respectively $\progress{c,A,\RewRules{}}$ and $\progress{I,d,\RewRules{}}$.
\end{definition}

\paragraph{Isolated sorts}

When there exists an active term $d$ for which rewrite progress does
not hold (with respect to its associated inert constants and rewrite
rules), the notion of canonicity, as defined in \cref{thm:canonicity}, is in danger.
To guarantee that in this case, canonicity in $\Type$ is still valid,
we rely on the notion of \emph{isolated sort}, which ensures that
there is no ``leak'' from the sort hosting $d$ into $\Type$.
In counterpart, well-formedness condition on eliminators from isolated
sorts needs to ensure that isolation is preserved.
In a well-formed signature $\Sigma$ the following invariant will be maintained:
%
%
\begin{align*}
  d \in \Sigma\enspace \wedge\enspace \isolated{s} \quad\implies\quad \isolated{s'} 
\end{align*}
where $d$ is an active constant and $\dom{d} = \univ{s} \vee \domsort{d} = s$ and $\cod{d} = \univ{s'} \vee \codsort{d} = s'$.



We can now turn to the definition of well-formed positive types,
negative types, eliminators (of positive types) and builders (of
negative types).
\begin{definition}[Well-formed positive type]
  \label{def:well-formed-positive-type}
  An inert constant $K$ building a type of sort $s$, $\cod{K} = \univ{s}$,
  together with its constructors $\Inert{K}$ is
  a well formed positive type in the signature $\Sigma$ when
  \begin{enumerate}
  \item There is no active constant $d \in \Sigma$ defined on $\univ{s}$
    \[\forall d \in \Active,\quad d \in \Sigma\enspace \Rightarrow\enspace \dom{d} \neq \univ{s}\]
  \item Its parameters are well-formed $\Jctx{\params{K}}$ and domain is empty
    $\dom{K} = \emptyContext$
  \item Any \emph{inert} constant $c \in \Inert{K}$ building a term in $K$, $\cod{c} = K (\overline{u})$, has
    parameters, domain and codomain well-formed in $\Sigma$:
    \begin{mathpar}
      \Jctx{\params{c}}
      \and
      \forall i, \dom{c}_i = K(\overline{t}) \wedge \Jsub{\params{c}}{\overline{t}}{\params{K}}
      \and
      \Jsub{\params{c}}{\overline{u}}{\params{K}}
  \end{mathpar}
  \end{enumerate}
\end{definition}
\noindent Rule (1) forces new type constructors of a sort to be checked before
eliminators on the universe. Rule (2) says that an inert type has only
well-formed parameters. Rule (3) checks that every inert constant
populating $K$ have well-formed parameters and strictly positive
occurrences of arguments in~$K$. 

\begin{definition}[Well-formed negative type]
  \label{def:well-formed-negative-type}
  An inert constant $K$ building a type of sort $s$, $\cod{K} = \univ{s}$,
  together with its observations $\Active{K}$ is a
  well formed negative type in the signature $\Sigma$ when
  \begin{enumerate}
  \item There is no active constant $d \in \Sigma$ defined on $\univ{s}$
    \[\forall d \in \Active,\quad d \in \Sigma\enspace \Rightarrow\enspace \dom{d} \neq \univ{s}\]
  \item Its parameters are well-formed $\Jctx{\params{K}}$ and domain is empty
    $\dom{K} = \emptyContext$
  \item $\Active{K}$ is an ordered set of \textbf{active} constant
    $\overline{d}$ that share the same parameters as $K$,
    $\forall i, \params{d_i} = \params{K}$.
    Any $d_i \in \Active{K}$ has domain $K$, $\dom{d} = K(\id_{\params{K}})$,
    and well-formed codomain in $\Sigma$ that can depend on the result of
    previous $\overline{d}_{<i}$:
    \[
      \begin{dcases*}
        \enspace\Jsub{\params{K}, \overline{\cod{d}}_{<i}}{\overline{t}}{\params{K}}& if $\cod{d_i} = K(\overline{t})$ \\[0.1cm]
        \enspace\Jty{\params{K}, \overline{\cod{d}}_{<i}}{\cod{d_i}}{\codsort{d_i}} & otherwise
      \end{dcases*}
    \]
  \item If $s$ is isolated then any $d \in \Active{K}$ land in an isolated sort:
    \[\forall s' \in \sorts, d \in\Active{d}, \isolated{s} \wedge (\cod{d} = \univ{s'} \vee
      \codsort{d} = s') \Rightarrow \isolated{s'}
    \]
  \end{enumerate}
\end{definition}
\noindent The two first rules are the same as for positive types, and
the last rule dually checks active constants.
More specifically, the condition that all active constants shares the
same set of parameters $\params{K}$ ensures that those active
constants actually define \emph{observations} of inhabitant of $K$.

\begin{definition}[Well-formed eliminator]
  \label{def:well-formed-elim}
  An active constant $d \in \Active{}$ together with rewrite rules
  $\RewRules{d}$ is well-formed in signature $\Sigma$ when
  the following conditions hold, where
  \[
    \begin{dcases*}
      s_d = s,\enspace \Inert{d} = \{\Pi\} \cup \{ K \in \Sigma \mid \cod{K} = \univ{s}\} &
      if $\dom{d} = \univ{s}$ is a universe \\
      s_d =
      \domsort{d},\enspace \Inert{d} = \Inert{K}& if $\dom{d} = K(\overline{u})$ is a
      positive type in $\Sigma$ 
    \end{dcases*}
  \]
  \begin{enumerate}
  \item If $s_d$ is isolated, $d$ must land in an isolated sort, otherwise it
    must satisfy progress:
    \[
      \begin{dcases*}
        (\cod{d} = \univ{s'} \vee \codsort{d} = s') \wedge \isolated{s'} & if~
        $\isolated{s_d}$\\
        \progress{d,\Inert{d},\RewRules{d}}&otherwise
      \end{dcases*}
    \]
  \item Its parameters, domain and
    codomains are well-formed in $\Sigma$
    \begin{mathpar}
      \progress{d, \Inert{d}, \RewRules{d}}\vee \isolated{s_d}
      \and
      \Jctx{\params{d}}
      \and
      \Jty{\params{d}}{\dom{d}}{\domsort{d}}
      \and
      \Jty{\params{d},\ctxvar{x}{\dom{d}}{\domsort{d}}}{\cod{d}}{\codsort{d}}
    \end{mathpar}
  \item The rewrite rules in $\mathcal{R}_d$ are deterministic $\determine{\RewRules{d}}$, have head symbol $d$ and are well-typed
    \begin{align*}
      \forall \rewrule{\Delta}{\sigma}{\Delta_{\lin{}}}{d}{\overline{x}}{pat}{r} \in
      \RewRules{d},\quad{}
      \rewtyping{\Sigma}{\Delta}{\sigma}{\Delta_{\lin{}}}{d}{\overline{x}}{pat}{r}
    \end{align*}
  \end{enumerate}
  \begin{mathpar}
  \end{mathpar}
\end{definition}

\begin{definition}[Well-formed builder]
  \label{def:well-formed-constr}
  An inert constant $c \in \Inert{}$ building a term of negative type $(K,
  \Active{K}) \in \Sigma$, that is $\cod{c} = K (\overline{u})$, together with
  rewrite rules $\RewRules{c}$ is well-formed
  in $\Sigma$ when:
  \begin{enumerate}
  \item It satisfies progress or belongs to an isolated sort and its parameters, domain and codomains are well-formed in $\Sigma$:
    \begin{mathpar}
      \progress{\Active{d}, c, \RewRules{c}} \vee \isolated{\codsort{c}}
      \and
      \Jctx{\params{c}}
      \and
      \forall i, \Jty{\params{c}, \dom{c}_{<i}}{\dom{c}_i} {\domsort{c}_i}
      \and
      \Jsub{\params{c}}{\overline{u}}{\params{K}}
    \end{mathpar}
  \item The rewrite rules in $\RewRules{c}$ are deterministic
    $\determine{\RewRules{c}}$, have patterns with head-constant $c$ and are well-typed
    \begin{align*}
      \forall \rewrule{\Delta}{\sigma}{\Delta_{\lin{}}}{d}{\overline{x}}{c(\overline{q})}{r} \in
      \RewRules{c},\quad
      \rewtyping{\Sigma}{\Delta}{\sigma}{\Delta_{\lin{}}}{d}{\overline{x}}{c(\overline{q})}{r}
    \end{align*}
  \end{enumerate}
\end{definition}

\begin{example}
  The type of lists, as other inductive types, is a positive
  type. Thus, to check that it can be added to a well-formed
  signature $\Sigma$, one need to check that $(\List{}, \{ \nilName,
  \consName\})$ is a well-formed positive types in $\Sigma$, and that
  $(\listrecName, \RewRules{\List{}})$ is a well-formed active
  constant in $\Sigma , (\List{}, \{ \nilName,
  \consName\})$.
  The typing conditions for well-formation of
  $(\List{}, \{ \nilName, \consName\})$ and
  $(\listrecName, \RewRules{\List{}})$ can be simply checked in
  general, as they are independent of $\Sigma$ in this specific situation,
  and have already been discussed in the presentation of this example.
  The last points to check are progress and determinism.
  This is fairly straightforward as it amounts to check that there is
  exactly one rule in $\RewRules{\List{}}$ for the eliminator
  $\listrecName$ and the constructors $\nilName$ and $\consName$.
\end{example}

\subsection{Metatheoretical properties of \MuTT}
\label{sec:metath-prop-mutt}
We will prove in \cref{sec:metatheory}
that the well-formedness of a signature is sufficient to prove that \MuTT{} enjoys the
following metatheoretical properties for any valid parametrization $\mathcal{P} =
(\sorts, \Sigma)$, making it well-suited as the underlying theory
of a proof assistant.

\begin{theorem}[Canonicity for the \Type{} hierarchy]\text{}
  \label{thm:canonicity}
  \begin{itemize}
  \item If $A$ is a closed type of sort $\Type$, $\Jty{}{A}{\Type}$, then $A$ is
    convertible to either a universe $\univ{s}$, a dependent product
    $\P{}(\ctxvar{x}{X}{s})Y$, or an inert constant $K(\overline{t})$.
  \item If $t$ is a closed term of a positive constant type $K(\overline{a})$ in \Type{},
    $\Jtm{}{t}{K(\overline{a})}{\Type{}}$ with $(K, \Inert{K}) \in \Sigma$, then
    $t$ is convertible to a
    constructor $c(\overline{p},\overline{v})$ with $c \in \Inert{K}$.
  \end{itemize}
\end{theorem}

Assuming that $\Sigma$ provides an empty type in $\Type{}$, that is a positive
type $\ofsort{\bot}{\Type{}}$ with no introduction $\Inert{\bot} = \varnothing$,
we also obtain the logical consistency of $\Type{}$ for any parametrization.

\begin{theorem}[Logical consistency of the \Type{} hierarchy]
  \label{thm:consistency}
  There is no closed proof term $e$ of the empty type $\Jtm{}{e}{\bot}{\Type}$.
\end{theorem}

\begin{theorem}[Decidability of conversion and typechecking]\text{}
  \label{thm:decidability}
  \begin{itemize}
  \item If $\Jty{\Gamma}{A}{s}$ and $\Jty{\Gamma}{B}{s}$, it is decidable
    whether $\Jconvty{\Gamma}{A}{B}{s}$ is derivable.
  \item If $\Jty{\Gamma}{A}{s}$, it is decidable whether $\Jtm{\Gamma}{t}{A}{s}$.
  \end{itemize}
\end{theorem}

\section{Expressivity and Instances of the Multiverse Type Theory}
\label{sec:examples}

This section develops several instances of parametrization of \MuTT, building
upon the formal framework described in \cref{sec:mutt}: inductive types
(\cref{sec:ind}), a \Coq-style sort of propositions (\cref{sec:example-prop}),
Exceptional Type Theory (\cref{sec:example-exctt}), axioms
(\cref{sec:example-ax}), and dependent elimination (\cref{sec:example-depelim}).


\subsection{Inductive and Record Types}
\label{sec:ind}
Section~\ref{sec:mutt} used lists to illustrate the definition of an inductive type; 
one can easily infer the definition of natural numbers $\Nat$ as a non-decorated version of lists.
We now show that standard $\Sigma$-types and identity types can also be
represented faithfully in \MuTT.

\paragraph{Σ types}
\ilabel{sec:sigma}
\def\SigmaName{\cstface{$\Sigma$}}
$\Sigma$ types represent dependent pairs. They are, like $\Pi$-types, 
a negative type constructor, defined by its two observations, the first and second projection.
\begin{itemize}[leftmargin=*,label=-]
\item The type constructor $\SigmaName$ in \Type{}
   is given by: $\params{\SigmaName} = \ctxvar{A}{\univ{\Type}}{\Type}_i, 
    \ctxvar{B}{A → \univ{\Type}}{\Type}_j$, \mbox{$\cod{\SigmaName} = \univ{\Type}_{\maxLevel{i}{j}}$}
  %
\item The projections are active constants $\fstName$ and $\sndName$ with
  $\params{\fstName} = \params{\sndName} = \params{\SigmaName}$, 
  $\dom{\fstName} = \dom{\sndName} = \SigmaName~A~B$ and
  $\cod{\fstName} = A$. For the second projection, $\cod{\sndName} = B~fst$,
  where $fst : A$ (see definition \ref{def:well-formed-negative-type}, item 3).
  This is an example where later projections depend on former ones.
\item The default builder constant $\pairName$ is an inert constant 
  presented by $\cod{\pairName} = \SigmaName~A~B$ and $\params{\pairName} = \params{\SigmaName}, \,
  \ctxvar{a}{A}{\Type} \ ,
  \ctxvar{b}{B~a}{\Type}$

\item We set $Δ = \params{\pairName}$ and 
  $\Delta_{\lin} = \params{\SigmaName}, \ctxvar{C}{\univ{\Type}}{\Type}, 
    \ctxvar{D}{A → \univ{\Type}}{\Type}, \ctxvar{c}{C}{\Type}, \ctxvar{d}{D~c}{\Type}$
  
  Ensuring progress we define the projection rewrite rules (deterministic because no overlap):
  \[\begin{array}{l}
    \rewtyping{\SigmaName,\fstName,\sndName}{\Delta}{A, B, A, B, a, b}
    {\Delta_{\lin}}{\fstName}{\metavar{A}, \metavar{B}}{\pairName(\metavar{C}, \metavar{D}, \metavar{c}, \metavar{d})}{c} \\
    \rewtyping{\SigmaName, \fstName, \sndName}{\Delta}{A, B, A, B, a, b}
    {\Delta_{\lin}}{\sndName}{\metavar{A}, \metavar{B}}{\pairName(\metavar{C}, \metavar{D}, \metavar{c}, \metavar{d})}{d}
  \end{array}\]
\item For typing purpose of the second rewrite rule: observe that (after the action of the $Δ_{\lin}$ to $Δ$ substitution) 
  $b$ has type $B~a$ according to the typing rule for \pairName. 
  This type is convertible to $B~(\fstName(A,B\rewcomma\pairName(A,B,a,b)))$ 
  thanks to the rewrite rule for $\fstName$.
\end{itemize}

\paragraph{Identity types}
\label{sec:identity}
\def\Id{\cstface{Id}\xspace}
\def\Refl{\cstface{refl}\xspace}
\def\J{\cstface{J}\xspace}

Illustrating the expressivity of our framework, we can also define standard Martin-Löf 
identity types \Id with the \J elimination rule of \citet{paulinTLCA93}.

\begin{itemize}[leftmargin=*,label=-]
\item $\params{\Id} = \ctxvar{A}{\univ{\Type}}{\Type}_i,\> \ctxvar{a}{A}{\Type},\> \ctxvar{x}{A}{\Type}$
  and $\cod{\Id} = \univ{\Type}_{i}$
  \item The unique constructor is $\Refl$, with $\params{\Refl} = \ctxvar{A}{\univ{\Type}}{\Type}_i, \ctxvar{a}{A}{\Type}$ and $\cod{\Refl} = \Id(A,a,a)$
  \item The elimination principle is an active constant $\J$ with $\dom{\J} = \Id(A,a,x)$, 
      $\cod{\J} = P~x~e$, and 
    \[
      \params{\J}  =  \params{\Id},\> \ctxvar{P}{Π (\ctxvar{x}{A}{\Type})~\Id(A,a,x){\to}\univ{\Type}}{\Type}, \>
      \ctxvar{\metavar{pr}}{P~a~(\Refl(A,a))}{\Type} 
    \]
  \item To define the rewrite rule for \J we set:
    \[\begin{array}{lcl}
      \Delta_{{\lin}} & = & \params{\J},\> \ctxvar{B}{\univ{\Type}}{\Type},\> \ctxvar{b}{B}{\Type} \\
      \Delta & = & \ctxvar{A}{\univ{\Type}}{\Type},\> \ctxvar{a}{A}{\Type},\>
                        \ctxvar{P}{Π (\ctxvar{x}{A}{\Type})~\Id(A,a,x){\to}\univ{\Type}}{\Type},\> \ctxvar{pr}{P~a~(\Refl(A,a))}{\Type}
      \end{array}\]
    Then we can introduce a well-typed rule:
    \[\begin{array}{l}
      \rewtyping{\Id,\Refl}{\Delta}{A, a, a, P, pr, A, a}
      {\Delta_{{\lin}}}{\J}{\metavar{A}, \metavar{a}, \metavar{x}, \metavar{P}, \metavar{pr}}{\Refl(B,b)}{\metavar{pr}}
    \end{array}\]
    Note here that the rewrite rule is highly non-linear: all the endpoints of
    equalities $a, x, b \in \Delta_\lin$ are enforced to coincide through typing.
\end{itemize}

Using propositional equality, other inductive families can be defined in the so-called ``Ford'' style 
\cite[\S 3.5]{mcbride99}, where proper indices are simulated by parameters and
equalities.
For example, to define vectors, we would have a type family with two parameters
$\ctxvar{A}{\Type}{\Type}, \ctxvar{n}{\Nat}{\Type}$ and the two constructors
would respectively be guarded by proofs of $\Id(\Nat,n,\Zero)$ and
$\Id(\Nat,n,\Succ~n')$.
A more categorically inspired, equivalent presentation of inductive families can
be found in \cite[\S 3.2]{conf/types/HerbelinS13} using Σ-types and identity
types to define proper indexed sums.
We expect higher-order recursive types like $W$-types could also fit in this
framework with a more elaborate handling of recursive occurrences
($\metavarrec{n}{p}$); we leave this as future work.
Beyond record types, generic coinductive types such as streams defined by
co-pattern matching~\cite{AbelPTS13} are almost within reach of \MuTT{}, in
particular the logical relation developed in~\cref{sec:metatheory} readily
accommodate them.
This would however require to dualize the treatment of reduction rules~\cref{fig:MuTT-patterns} and their
typing~(Rule \nameref{infrule:rew-typing-term}) to take into account co-recursive occurrences.

\subsection{Prop}
\label{sec:example-prop}
The \Coq proof assistant features a sort \Prop{} (\verb|Prop|) of propositions compatible with
a proof erasure semantics, a key property for extracting formally verified programs from \Coq developments.
This compatibility is obtained through a restricted elimination schema from
$\prop$ into $\type$, known as {\em singleton elimination}, that enforces that
\Prop is compatible with proof-irrelevance but does not impose that axiom
upfront.
Singleton elimination, first studied explicitly by \citet{letouzey:phd}, restricts
elimination on inductive types from \Prop to \Type{} to those inductives that have 
syntactically at most one constructor and whose arguments are all in the sort \Prop.
As a consequence, the standard proof to distinguish the constructors of an
inductive type, e.g. to distinguish \true{} from \false{} in $\bool{}$, cannot
be reproduced for types in \Prop, e.g. $\boolProp{}$: the first step of the
proof builds a predicate $P : \boolProp{} \to \univ{\Type{}}$ such that
$P\,\true$ is inhabited and $P\,\false$ is empty, a step that requires the
forbidden elimination of $\ofsort{\boolProp{}}{\Prop}$ into the universe
$\ofsort{\univ{\Type}}{\Type}$.
Putting aside the peculiar aspects attached to impredicativity, that indeed turn \Prop{}
into a hierarchy on its own, it is relatively straightforward to encode a
predicative variant of \Coq-style \Prop{} in \MuTT{}.
\newcommand\I{\cstface{I}}
\newcommand\ccst{\cstface{c}}
We consider a fresh sort \Prop{}, and populate it with inductive type formers
$\I(\overline{p})$ and constructors $\ccst(\overline{x}) :
\I(\overline{p})$ as described in \cref{sec:ind} but restrict the usual
eliminators $\cstface{IRec}$ to \Prop-valued families, that is predicates $P :
\I(\overline{p}) \to \univ{\Prop}$.
Eliminators $\cstface{singletonRec}_{\I}$ for $\Type$-valued families are added only for
inductives $\I(\overline{p})$ that do satisfy the singleton criterion.



\subsection{Exceptions}
\label{sec:example-exctt}
\newcommand{\elimsort}{s}

Exceptional Type Theory (\ExcTT{}) of~\cite{pedrotTabareau:esop2018}
extends \MLTT{} with the ability to raise exceptions at any type.
This theory is further refined in~\cite{pedrot19} in order to reason on
exceptional terms in a consistent context.
This is achieved by introducing two sorts: a sort of pure types embedding standard \MLTT{}
and a sort \exc{} for the exceptional hierarchy.
The use of exceptions is confined to types residing in the exceptional
hierarchy, which is explicit in the typing rule for the operator \raiseC{}
raising those exceptions:
\[
  \inferrule 
  {\Jty{\Gamma}{A}{\exc}}{\Jtm{\Gamma}{\raiseFull{A}}{A}{\exc}}
\]
The eliminator for inductive types in \exc{} must then account for these
exceptions, requiring a special catch clause.
For instance, the eliminator from exceptional booleans \boolE to a sort
\(\elimsort\) takes the following form:
\[
    \inferrule 
    {\Jtm{\Gamma}{P}{\boolE \to \univ\elimsort}{\Type}\\
      \Jtm{\Gamma}{b}{\boolE}{\exc}\\
     \Jtm{\Gamma}{h_{t}}{P~\true}{s}\\
     \Jtm{\Gamma}{h_{f}}{P~\false}{s}\\
     \Jtm{\Gamma}{h_{r}}{P~(\raiseFull{\boolE})}{s}}
    {\Jtm{\Gamma}{\catchBool{P}{h_{t}}{h_{f}}{h_{r}}{b}}{P~b}{s}}
\]
The specification of this eliminator is completed with the reduction rules where
$\overline{p} = P,b,h_{t},h_f,h_r$:
\begin{align}
  \label{eq:catchB-rew}
  \acst{\catchB}{\overline{p}}{\true} &\Rightarrow h_t&
  \acst{\catchB}{\overline{p}}{\false} &\Rightarrow h_f&
  \acst{\catchB}{\overline{p}}{\raiseFull{\boolE}} &\Rightarrow h_r
\end{align}
%

%
We can present \ExcTT as an instance of \MuTT, reusing the sort \Type{} for pure
types plus a fresh sort $\exc{} \in \sorts$ populated with the inert constants
$\boolE, \true, \false, \boolexn$ and the active constants $\catchB$ and
$\raiseC$.
%
%
%
%
The exceptional booleans $\boolE$ with $\cod{\boolE} = \univ\exc$ have as
constructors, beside the standard \true and \false constructors, a new
constructor $\boolexn$, all without parameters.
This makes \(\boolE\) a well-formed positive type in the empty signature and
\(\catchB\) is introduced as an eliminator on this well-formed type with:
\begin{small}
\begin{mathpar}
    \params{\catchB} = \ctxvar{P}{\boolE \to \univ\elimsort}{\Type},\>
    \ctxvar{h_{t}}{P\,\true}{\elimsort},\> \ctxvar{h_{f}}{P\,\false}{\elimsort},\>
    \ctxvar{h_{r}}{P\,\boolexn}{\elimsort}
    \\
    \dom{\catchB} = \boolE
    \and
    \Jnosigty{\boolE}{\params{\catchB},\ctxvar{b}{\dom{\catchB}}{\exc}}{\cod{\catchB}
    = P\,b}{\elimsort}
\end{mathpar}
\end{small}
together with the linear equations presented in~(\ref{eq:catchB-rew}).
Again, \(\catchB\) is well-formed in the signature \(\boolE\) because 
    it is deterministic, as all patterns have different head symbols, and 
    satisfies rewrite progress, as all non-neutral weak-head normal
        forms of type $\boolE$ (\ie~\(\true, \false\) and \(\boolexn\)) do
        react. Also, parameters, domain, codomain and rewrite
        rules are well-typed  with respect to the signature \(\boolE\).
%
Finally, \(\raiseC\) is presented as an eliminator with \(\params{\raiseC} =
\emptyContext\), hence abbreviated $\raiseC(A)$ instead of
$\raiseC(\cdot\rewcomma A)$, defined on the universe \(\dom{\raiseC} := \univ\exc\), with codomain
\(\cod{\raiseC} = A, \codsort{\raiseC} = \exc\) where
$\oftype{A}{\univ\exc}{\Type}$ is the variable provided by the domain.
\raiseC{} then comes with a rewrite rule for \(\P\)-types and another one for
booleans: 
\begin{mathpar}
    \rewrulelinx{\emptyContext}{\raiseC}{\boolE}{\boolexn}
    \and
    \rewrulelinx{(\ctxvar{A}{\univ{s_1}}{\Type},\>\ctxvar{B}{A\to
        \univ\exc}{\Type})}{\raiseC}{\P^{s_1}\metavar{A}~
    (\metavarrec{B}{\emptySub})}{B^{\rec}}
\end{mathpar}

The well-formedness of \(\raiseC\) is established as follows.
Determinism comes from the fact that all rewrite rules have a distinct head
symbol in their pattern.
All non-neutral weak-head normal forms in \(\univ\exc\) (\(\boolE\) and \(\P
A~B\)) do react.
The parameters, domain and codomain of \(\raiseC\) are well-typed.
Regarding the typing of the rewrite rules, the first one is easily well-typed.
As for the second, the \(\P\) pattern is indeed well-typed given the context,
and the variable \(B^{\rec}\) comes from \(\occrec{\raiseC}{\P^{s_1} \metavar{A}
  \metavarrec{B}{\emptySub}}\) and has type \(\P(\ctxvar{y}{A}{s_1}) (B~y)\).

Note that, if we wanted to extend a base signature with additional type
constructors from \cref{sec:ind} like $\Sigma$-types or identity types,
the framework would require to consider these type constructors as
additional inert constants on which $\raiseC$ should react.

%

\subsection{Axioms, locally}
\label{sec:example-ax}

\newcommand\axiomTm[0]{\cstface{axiom}}

By parametrizing adequately \MuTT, it is possible to add and work with a new
axiom inhabiting any chosen type \(Ax\) without compromising the
canonicity of $\Type$.
This is achieved by creating a new isolated sort \(\Ax\) which is a fresh
copy of \(\Type\), except that eliminations are restricted to \(\Ax\).
Then, the axiom can be realized by adding an active term
\(\oftype{\axiomTm}{Ax}{\Ax}\)
with no parameters ($\params{\axiomTm} = \emptyContext$) and $\dom{d} =
\Unit$, the trivial inductive type (or equivalently, no argument at all).
We do not attach any rewrite rule to $\axiomTm$ as this would amount
to realizing the axiom itself, breaking progress and therefore, $\axiomTm$ is well-formed only
because $\Ax$ is isolated, which precisely prevents leaking the axiom
into $\Type$.

However, the isolation property does not prevent us from defining a
boxing mechanism from $\Type$ into $\Ax$ with elimination into $\Ax$
that allows us to prove properties on inhabitants of $\Type$ using
$\axiomTm$, but only in the axiomatic sort $\Ax$.
Then, depending on the design choice, one can define one axiomatic
sort per axiom to encapsulate clearly which axiom has been used
directly in the type information, or consider an axiomatic sort
where any axiom can be postulated, thus encapsulating in the types
that an unsafe version of $\Type$ has been used.

This provides a type-theoretic, local and modular alternative to the \texttt{--safe}
pragma of \Agda, or the \texttt{Print Assumption} checker of \Coq.
Also, it allows users to make use of several incompatible axioms in the same
development, as long as they are postulated in different isolated
sorts.

\subsection{Dependent elimination through universe unboxing}
\label{sec:example-depelim}

A sort $s \in \sorts$ has booleans if it is equipped with a type
$\ofsort{\bool_s}{s}$, terms $\oftype{\true_s,\false_s}{\bool_s}{s}$ and an
induction principle
\[\ind_{\bool_s} : (\ctxvar{P}{\bool_s \to
  \univ{s}}{\Type{}})(\ctxvar{p_t}{P\,\true_s}{s})(\ctxvar{p_f}{P\,\false_s}{s})(\ctxvar{b}{\bool_s}{s})
\to P\,b.
\]
As explained in \cref{sec:example-prop}, this data is however not enough to show
expected properties of booleans, for instance to derive that $\true_s \not\equiv
\false_s$.
In order to recover the full power of large elimination on booleans of sort
$s$, we need the ability to define predicates taking value in $\univ{s}$ by case analysis:
\[
  \inferrule{\Jtm{\Gamma}{P_{t}}{\univ{s}}{\Type{}}
    \and \Jtm{\Gamma}{P_{f}}{\univ{s}}{\Type{}}
    \and \Jtm{\Gamma}{b}{\bool_s}{s}}
  {\Jtm{\Gamma}{\ind^{\univ{s}}_{\bool_s}\,P_t\,P_f\,b}{\univ{s}}{\Type{}}}
\]
Note that this induction principle $\ind^{\univ{s}}_{\bool_s}$ specialized to
$\univ{s}$ is not an instance of $\ind_{\bool_s}$ because
$\univ{s}$ resides in the sort $\Type{}$.
Rather than requiring for each inductive type in sort $s$ to come equipped with
two elimination principles, it is actually enough to have a reflection of
$\univ{s}$ in sort $s$, that is a type $\ofsort{\BoxInd\,\univ{s}}{s}$ equipped
with terms $\Jtm{}{\boxC}{\univ{s} \to \BoxInd\,\univ{s}}{s}$ and
$\Jtm{}{\unbox}{\BoxInd\,\univ{s} \to \univ{s}}{\Type{}}$ such that
$\Jconvtm{\Gamma}{\unbox \,(\boxC\,A)}{A}{\univ{s}}{\Type{}}$ for any type $\oftype{A}{\univ{s}}{\Type{}}$.
Given such a reflection, we can derive the induction principle
$\ind^{\univ{s}}_{\bool_s}$ from the standard induction principle as follows:
\[\ind^{\univ{s}}_{\bool_s}~P_t~P_f~b\eqbydef \oftype{\unbox \left (\ind_{\bool_s} (\lambda
  (x : \bool_s).\,\BoxInd\,\univ{s})~(\boxC\,P_t)~(\boxC\,P_f)\,b \right )}{\univ{s}}{\Type{}}\]

Of course, such a reflection does not always exist for an arbitrary
sort $s$, in particular in the case of $\Prop$. But it exists for
instance for $\exc$ which justifies why dependent elimination is valid
in \ExcTT. The term $\ofsort{\BoxInd\,\univ{\exc}}{\exc}$ is basically
obtained as $\univ{\exc}$ plus a default type $\deamon$ in $\exc$ for the exception in
$\BoxInd\,\univ{\exc}$, \ie $\raiseFull{\BoxInd\,\univ{\exc}} \equiv \deamon$.

\section{Modularity of \MuTT{}}
\label{sec:combination}

As formalized and illustrated previously, \MuTT is extensible by means of its parametrization. 
We now show that \MuTT delivers on the {\em modularity} front: two parametrizations 
of \MuTT can be merged seamlessly, preserving the metatheoretical results of each independent parametrization. 
As explained in the introduction, modularity is key to allow developments to
locally rely on extensions such as exceptions or axioms, without interfering
with each other, i.e. the metatheoretical properties in~\cref{sec:metatheory}
are always preserved.
%

%
A parametrization $\mathcal{P}' = (\sorts', \Sigma')$ is a proper extension of
$\mathcal{P} = (\sorts, \Sigma)$, noted $\mathcal{P} \extendMuTT{}
\mathcal{P}'$, when $\sorts \subseteq \sorts'$, each component of $\Sigma$ are
in $\Sigma'$, isolation is preserved and any active constant $d \in \Sigma'$ defined on a universe $s$ is
either already in $\Sigma$ or occurs on a sort not appearing in \sorts:
\[\forall d \in \Sigma',\qquad \dom{d} = \univ{s} \quad \Rightarrow\quad d \in \Sigma \vee s
  \notin \sorts\]
For any parametrization $\mathcal{P}$, the identity is a proper extension $\mathcal{P} \extendMuTT{} \mathcal{P}$
and proper extensions compose, forming a preorder with initial object
$\mathcal{P}_\Type = (\{\Type\}, \emptySig{})$.
\begin{lemma}[Functoriality]
  All typing judgments of \MuTT{} presented in \cref{fig:MuTT-judgments} are
  functorial with respect to proper extensions, that is, if $(\sorts, \Sigma)
  \extendMuTT{} (\sorts', \Sigma')$ and $\Sigma ; \Gamma \vdash \mathcal{J}$ is
  a derivable judgment then $\Sigma'; \Gamma \vdash \mathcal{J}$ is also derivable.
\end{lemma}
\begin{proof}
  By induction on the derivation of the judgment $\Sigma ; \Gamma \vdash
  \mathcal{J}$ noting that typing derivations only use that constants belongs to
  $\Sigma$, a property preserved by proper extensions.
\end{proof}
Importantly, well-formed signature extensions, as defined in
\cref{def:well-formed-positive-type,def:well-formed-negative-type,def:well-formed-elim,def:well-formed-constr},
are compatible with proper extensions, with the exception of active constants
defined on universes:
\begin{lemma}\label{lem:sigexts-proper}
  Suppose $(\sorts, \Sigma) \extendMuTT{} (\sorts', \Sigma')$.
  \begin{description}
  \item[Positive type extension] If $\Sigma, (K, \Inert{K})$ is a well-formed signature then so is $\Sigma', (K, \Inert{K})$;
  \item[Negative type extension] If $\Sigma, (K, \Active{K})$ is a well-formed signature then so is $\Sigma', (K, \Active{K})$;
  \item[Eliminator extension] If $\Sigma, (d, \RewRules{d})$ is a well-formed
    signature with $\dom{d} = K(\overline{u})$ then so is $\Sigma',
    (d, \RewRules{d})$;
  \item[Builder extension] If $\Sigma, (c, \RewRules{c})$ is a well-formed signature then
    so is $\Sigma', (c, \RewRules{c})$.
  \end{description}
\end{lemma}

\begin{proof}
Suppose $(K, \Inert{K})$ is a well-formed positive type in $\Sigma$ with
$\cod{K} = \univ{s}$, we show that it is a well-formed positive type in $\Sigma'$.
By functoriality of judgments, all conditions but the first are satisfied.
Suppose $d \in \Sigma'$ has domain a universe $\univ{s'}$ then either $d \in
\Sigma$ and, by well-formedness of $K$, $s' \neq s$, or $s' \notin \sorts$ so
$s' \neq s$.
A similar argument applies for well-formed negative types.
Suppose $(d, \RewRules{d})$ is a well-formed eliminator on a positive constant,
$\dom{d} = K(\overline{u})$ in $\Sigma$, we show that it is well-formed as well in $\Sigma'$.
Since $(K, \Inert{K}) \in \Sigma$ and proper extensions preserve components,
$(K, \Inert{K}) \in \Sigma'$.
$\determine{\RewRules{d}}$ and $\progress{\Inert{K},d,\RewRules{d}}$ are
independent from the signature and all the typing conditions are consequences of
functoriality along proper extensions.
If $\domsort{d}$ is isolated in $\Sigma$ then $\codsort{K}$ is isolated in
$\Sigma$, and they are both isolated in $\Sigma'$ because
proper extensions preserve isolation.
Again, a similar argument applies for well-formed builders.
\end{proof}


%

\begin{theorem}[Combining parametrizations]
\label{th:comb-param}
  Let $\mathcal{P} = (\sorts, \Sigma)$ be a parametrization of \MuTT{} and
  $\mathcal{P}_1 = (\sorts_1, \Sigma_1)$, $\mathcal{P}_2 = (\sorts_2, \Sigma_2)$
  two proper extensions of $\mathcal{P}$. There exists a well-formed
  signature $\Sigma_\cup$ on $\sorts_{\cup} = \sorts_1 \uplus_\sorts \sorts_2$
  such that $\mathcal{P}_\cup = (\sorts_{\cup}, \Sigma_\cup)$ is a proper
  extension of both $\mathcal{P}_1$ and $\mathcal{P}_2$.
\end{theorem}
\begin{proof}
Without loss of generality, we can assume that $\Sigma$ is a prefix of $\Sigma_1
= \Sigma \Sigma'_1$.
%
The construction of $\Sigma_\cup$ then proceeds by induction on $\Sigma'_1$,
extending inductively the proper extension $\mathcal{P} \extendMuTT{}
\mathcal{P}_2$.
All cases are covered by \cref{lem:sigexts-proper}, except for the case of a well-formed
active constant with domain a universe, which we now discuss.
Assume that $(d, \RewRules{d})$ is well-formed in $\Sigma \Sigma'_1$, $\dom{d} = \univ{s}$.
The induction hypothesis states that we have well-formed signatures
$(\sorts_1,\Sigma\Sigma'_1) \extendMuTT{} (\sorts_{\cup},
\Sigma_2 \Sigma'_1)$ together with the indicated proper extension.
Since $d \in \Sigma_1$ is not part of the prefix $\Sigma$, by properness of $\mathcal{P}
\extendMuTT{} \mathcal{P}_1$, $s \notin \sorts$ so $s \notin \sorts_2$ and the
signatures $\Sigma\Sigma'_1$ and $\Sigma_2\Sigma'_1$ introduce the same
constants in the universe $\univ{s}$. Therefore, if $\progress{\Inert{d},d,\RewRules{d}}$
holds for $\Inert{d}$ computed in $\Sigma\Sigma'_1$
it also holds for $\Inert{d}$ computed in $\Sigma_2\Sigma'_1$.
The other conditions to show that $(d,\RewRules{d})$ is well formed in
$\Sigma_2\Sigma'_1$ hold by functoriality of typing judgments, preservation of
isolation and independence with respect to the signature for
$\determine{\RewRules{d}}$ thus concluding the inductive step.
By construction, we have both $(\sorts_2,\Sigma_2) \extendMuTT{}
(\sorts_{\cup}, \Sigma_{\cup})$ and $(\sorts_1,\Sigma_1) \extendMuTT{}
(\sorts_{\cup}, \Sigma_{\cup})$.
\end{proof}

As a crude 
application of \cref{th:comb-param}, we can combine almost disjoint
parametrizations that agree on the sort $\Type$:
\begin{corollary} If $\mathcal{P}_1 = (\sorts_1,\Sigma_\Type\Sigma_1)$
  and $\mathcal{P}_2 = (\sorts_2, \Sigma_\Type\Sigma_2)$ are \MuTT{}
  parametrization that only share $\sorts_1 \cap \sorts_2 = \{\Type\}$, and
  agree on a common signature $(\{\Type\}, \Sigma_\Type)$, then $(\sorts_1\cup
  \sorts_2,  \Sigma_\Type\Sigma_1\Sigma_2)$ is a valid \MuTT{}
  parametrization.
\end{corollary}

In other words, combined with the metatheoretical results presented in
\cref{sec:metath-prop-mutt}, this corollary shows that \MuTT addresses
the logical modularity issue depicted in the introduction of this
paper.
Combining two valid \MuTT parametrizations yields to a consistent type
theory, even if the logical principles provided in those theories
are not compatible altogether.
This logical frontier has been achieved by the multiverse setting that
allows to localize in a sort the use of new logical principles, and
also the use of their consequences, which may not even mention
explicitly those new principles. 

\section{Metatheory of \MuTT}
\label{sec:metatheory}

In this section, we show the metatheoretical properties of \MuTT claimed in
\cref{sec:mutt} by adapting and extending the mechanized logical relation proof
of~\citet{AbelOV18}.
The high-level idea of the proof is standard: we carefully define a logical
relation exhibiting for each derivable judgments of \MuTT{} a canonical standard
shape for its derivation, show that the resulting logical relation satisfies a
variety of properties then prove the fundamental lemma by induction on typing
derivations.
Finally we derive the actual metatheoretical properties as consequences of the
fundamental lemma.

\subsection{Logical relation}
\begin{figure}
\begin{small}
  \begin{tabular}{ll}
    $\LRctx{\Gamma}{\Delta}$ & $\Delta$ is a reducible telescope on top of $\Gamma$ with respect to signature $\Sigma$\\
    $\LRsub{\Gamma}{\sigma}{\Delta}{[\Delta]}$ & $\sigma$ is a reducible substitution to the extension $\LRarg{[\Delta]} : \LRctx{\Gamma}{\Delta}$ \\
    $\LRty{\Gamma}{A}{s}$& $A$ is a reducible type at sort $s \in \sorts$ in context $\Gamma$\\
    $\LRtm{\Gamma}{t}{A}{s}{[A]}$& $t$ is a reducible term of reducible type $\LRarg{[A]} : \LRty{\Gamma}{A}{s}$\\
    $\LRconvty{\Gamma}{A}{B}{s}{[A]}$& $B$ is convertible to the reducible type $\LRarg{[A]} : \LRty{\Gamma}{A}{s}$\\
    $\LRconvtm{\Gamma}{t}{u}{A}{s}{[A]}$& $t$ and $u$ are convertible at reducible type $\LRarg{[A]} : \LRty{\Gamma}{A}{s}$ \\[0.4cm]
    $\LRneutraltm{\Gamma}{t}{A}{s}$ & $t$ is a reducible neutral term of type $A$
  \end{tabular}
  \end{small}
  \caption{Components of the logical relation for \MuTT{}}
  \label{fig:MuTT-reducibility}
\end{figure}
The logical relation defines families of types, the \emph{reducibility}
relations $\Vdash$ described in \cref{fig:MuTT-reducibility}, corresponding to
each judgment of \MuTT{}~(\cref{fig:MuTT-judgments}).
The definition of these relation proceed first by induction on the well-formed
signature $\Sigma$, collecting inductively proof of reducibility data employed
to define the reducibility relation at types introduced by inert constants $K$
(\cref{def:well-formed-positive-type,def:well-formed-negative-type}).
We use the notation $\LRarg{[x]}$ for the proof of reducibility associated to a
type, term, context or substitution $x$.
We abuse application notation $\LRarg{[A]~[t]}$ to substitute a reducibility proof $[t]
: \LRtm{\Gamma}{t}{X}{s}{[X]}$ in $\LRarg{[A]} : \Sigma; \Gamma,
\oftype{x}{X}{s} \Vdash \mathcal{J}$, omitting the required substitutions and
providing only the main arguments.
\begin{definition}[Reducibility of positive type constant]
  A positive type constant $(K, \Inert{K})$ is reducible in signature $\Sigma$ if
  \begin{enumerate}
  \item its parameters are reducible $\LRarg{[\params{K}]} : \LRctx{\emptyContext}{\params{K}}$
  \item for each inert constant $c \in \Inert{K}$, $\cod{c} =
    K(\overline{u})$, $\params{c}$, $\dom{c}$ and $\overline{u}$ are reducible
    \begin{mathpar}
      \LRarg{[\params{c}]} : \LRctx{\emptyContext}{\params{c}}
      \and
      \LRarg{[\overline{u}]} : \LRsub{\params{c}}{\overline{u}}{\params{K}}{[\params{K}]}
      \and
      \forall i,\enspace{} \dom{c}_i = K(\overline{t}) \wedge \LRarg{[\dom{c}_i]} : \LRsub{\params{c}}{\overline{t}}{\params{K}}{[\params{K}]}
    \end{mathpar}
  \end{enumerate}
\end{definition}
\begin{definition}[Reducibility of negative type constant]
  A negative type constant $(K, \Active{K})$ is reducible in signature $\Sigma$ if
  \begin{enumerate}
  \item its parameters are reducible $\LRarg{[\params{K}]} : \LRctx{\emptyContext}{\params{K}}$
  \item for each active constant $d_i \in \Active{K} = \overline{d}$, , $\cod{d_i}$ is reducible
    \[
      \begin{dcases*}
        \enspace
        \LRarg{[\cod{d_i}]} : \LRsub{\params{K}, \overline{\cod{d}}_{<i}}{\overline{t}}{\params{K}}{[\params{K}]}
        & if $\cod{d_i} = K(\overline{t})$ \\[0.1cm]
        \enspace
        \LRarg{[\cod{d_i}]} : \LRty{\params{d},  \overline{\cod{d}}_{<i}}{\cod{d_i}}{\codsort{d_i}}
        & otherwise
      \end{dcases*}
    \]
  \end{enumerate}
\end{definition}
\begin{definition}[Reducibility of rewrite rules]
  A rewrite rule
  $\rewrule{\Delta}{\sigma}{\Delta_{\lin{}}}{d}{\overline{x}}{pat}{r}$ is
  reducible in signature $\Sigma$ if its contexts, substitutions and
  right hand side are reducible:
  \begin{mathpar}
    \LRarg{[\Delta]} : \LRctx{\emptyContext}{\Delta}
    \and
    \LRarg{[\Delta_{\lin}]} : \LRctx{\emptyContext}{\Delta_\lin}
    \and
    \LRctx{\Delta{}_\lin}{\occrec{d}{pat}}
    \and
    \LRsub{\Delta{}}{\sigma{}}{\Delta{}_\lin}{[\Delta{}_\lin]}
    \and
    \LRsub{\Delta{}_\lin}{\overline{x}}{\params{d}}{[\params{d}]}
    \and
    \LRtm{\Delta{}_\lin}{\erasepat{pat}}{A}{\domsort{d}}{[A]}
    \and
    \LRtm{\Delta{}, \occrec{d}{pat}[\sigma]}{r}{\cod{d}[\overline{x}, \erasepat{pat}][\sigma{}]}{\codsort{d}}{[\cod{d}]~\ldots}
  \end{mathpar}
\end{definition}
Assuming by induction hypothesis that we have the reducibility datum $\LRarg{[\Sigma]}$
corresponding to a well-formed signature $\Sigma$, we can now describe the
defining cases of the logical relation.
In all cases, the general methodology is to first reduce the subject of the
judgment to a whnf, and then characterize the canonical forms at each type and
judgments.
The definition is complicated by two aspects exposed in~\cite{AbelOV18}: first,
universes introduce a seemingly circular definition of reducibility at terms and
types; second, negative occurrences of the reducibility relation appears in the
definition, e.g. for reducibility of terms at dependent products.
The circularity induced by universes actually disappears once we take into
account the (implicit) universe levels.
The problem induced by negative occurrences is solved by defining inductively the
reducibility relation on types, and then defining the other relations by
recursion on the proofs of reducibility on types.
Since the definition of all these relations are mutual, we obtain a 
well-founded albeit complex inductive recursive definition.
\Cref{fig:red-types} describes the inductive part of reducibility of types with
a case for universes of each sorts, neutral types, $\Pi$-types that use the
context reducibility from \cref{fig:red-context} and constant introduced from
the signature $\Sigma$.
\Cref{fig:red-terms} then dispatches reducibility of terms to auxiliaries
definitions according to the reducibility proof of its type. 
We omit most of these auxiliary definitions, focusing on the components proper to
\MuTT{} and absent from~\cite{AbelOV18}.
\Cref{fig:red-constant} describes the reducibility of terms at types introduced
with a constant $K$ drawn from the signature $\Sigma$.
The rule \nameref{infrule:constant-reducible} is the entry point and closes the
other judgments by weak head reduction, while the other
rules apply depending on the positive or negative character
of the type $K$ as described by the signature $\Sigma$.
When $(K,\Inert{K})$ is a positive type constant according to $\Sigma$, a whnf
is reducible at $K(\overline{a})$, if it is neutral using
\nameref{infrule:neutral-reducible}, or if it consists of an inert constant $c
\in \Inert{K}$, one of the canonical introduction form of $K$, its parameters
and recursive arguments are \textbf{inductively} reducible and the arguments of $K$ computed from its
parameters are convertible to $\overline{a}$, using rule
\nameref{infrule:inert-constant-positive-reducible}.
When $(K, \Active{K})$ is a negative type constant according to $\Sigma$, a whnf
$t$ is reducible at $K(\overline{a})$, rule
\nameref{infrule:inert-constant-negative-reducible}, if all its possible
observations $d(\overline{a},t)$ for $d \in \Active{K}$ are
\textbf{coinductively} reducible at their corresponding type.
In both of these cases, the definitions make a crucial use of the reducibility
data $\LRarg{[\Sigma]}$ obtained by induction hypothesis on $\Sigma$.

The so-defined logical relation verifies a handful of properties:
\begin{enumerate}
\item it is stable under weakening, substitution by reducible substitution;
\item the relations induced by conversion are reflexive, symmetric, transitive
  and congruent with respect to all type and term formers;
\item all reducibility relations are stable by judgmental conversion;
\item a reducible type or term reduces to a whnf that is itself reducible;
\item reducibility is closed by anti-reduction;
\item well-typed neutrals are reducible.
\end{enumerate}

We highlight two key properties: the reducibility relations are \emph{irrelevant},
so that being reducible is a mere property; and all judgments satisfy the so
called escape lemma that allows to recover derivability of a \MuTT{} judgment
out of its reducible counterpart. 

\begin{lemma}[Irrelevance]
  If $\LRarg{[A]}, \LRarg{[A']} : \LRty{\Gamma}{A}{s}$ are two proofs of
  reducibility of the type $A$, $\LRtm{\Gamma}{t}{A}{s}{[A]} \Rightarrow
  \LRtm{\Gamma}{t}{A}{s}{[A']}$.
\end{lemma}
The key property of \MuTT{} needed to prove irrelevance of the reducibility witnesses 
is the determinism of rewrite rules, ensuring uniqueness of the head of a weak head normal form.
\begin{lemma}[Escape]
  For any judgment form $\mathcal{J}$ of \MuTT{}, if $\Sigma ; \Gamma \Vdash
  \mathcal{J}$ then $\Sigma ; \Gamma \vdash \mathcal{J}$.
\end{lemma}
The escape lemma reconstructs a canonical derivation of a judgment out of a
reducibility proof.
Irrelevance is used pervasively to ``realign'' reducibility judgments that only
differ in the reducibility proof.

\begin{figure}
\begin{small}
  \flushleft

  \fbox{$\LRctx{\Gamma}{\Delta}$} \emph{reducibility of telescopes.}
  \begin{mathpar}
    \inferrule{ }{\LRctx{\Gamma}{\emptyContext}} 
    \and
    \inferrule{
      \LRarg{[\Delta{}]} : \LRctx{\Gamma{}}{\Delta{}} \\
      \LRty{\Gamma{}, \LRpair{\Delta}{[\Delta]}}{A}{s}
   }{\LRctx{\Gamma{}}{\Delta{}, \ctxvar{x}{A}{s}}}
 \end{mathpar}
 \[
   \begin{array}{rcr}
     \LRty{\Gamma{}, \LRpair{\Delta}{[\Delta]}}{A}{s} &\eqbydef&
     \forall \Gamma{}' \supseteq \Gamma{},\enspace
     \LRsub{\Gamma{}'}{\sigma{}}{\Delta{}}{[\Delta{}]}\to
       (\LRarg{[A\sigma]} : \LRty{\Gamma{}'}{A[\sigma]}{s} \enspace \wedge\\
     &&(\LRsub{\Gamma{}'}{\sigma{}'}{\Delta{}}{[\Delta{}]}\to
       \enspace{}\LRconvtm{\Gamma{}'}{\sigma{}}{\sigma{}'}{\Delta{}}{}{[\Delta{}]}\to\\
                                                      &&\LRconvty{\Gamma{}'}{A[\sigma{}]}{A[\sigma{}']}{s}{[A\sigma]}))\\
     \Delta \supseteq \Gamma&\eqbydef & \exists \rho, \Delta \vdash \rho \Gamma
     \wedge \rho \text{ is a monotone renaming}\hfill{}
     
   \end{array}
 \]%
  
 \fbox{$\LRsub{\Gamma}{\sigma}{\Delta}{[\Delta]}$} \emph{reducibility of substitutions.}
 \begin{mathpar}
   \inferrule{ }{\LRsub{\Gamma}{\emptySub}{\emptyContext}{[\emptyContext]}} 
   \and
   \inferrule{\LRarg{[\sigma{}]} :\LRsub{\Gamma}{\sigma}{\Delta}{[\Delta]}\\
    \LRtm{\Gamma{}}{t}{A[\sigma{}]}{s}{([A]~id_{\Gamma{}}~[\sigma]).1}
   }{\LRsub{\Gamma}{(\sigma, t)}{\Delta, \ctxvar{x}{A}{s}}{[\Delta], [A]}}
 \end{mathpar}

 \end{small}
  \caption{Reducibility of context and substitutions}
  \label{fig:red-context}
\end{figure}

  

%

\newcommand\neuEqRel[5]{\Sigma; #1 \vdash \oftype{#2 \sim #3}{#4}{#5} }
\newcommand\eqRel[5]{\Sigma; #1 \vdash \oftype{#2 \cong #3}{#4}{#5} }

\newcommand\LRRedTyUniv[1]{[\univ{#1}]}
\newcommand\LRRedTyNe[0]{[\mathrm{ne}]}
\newcommand\LRRedTyPi[0]{[\Pi]}
\newcommand\LRRedTyK[0]{[\mathrm{Cst}]}

\begin{figure}
\begin{small}
  \begin{mathpar}
    \inferrule[Univ-Red-Type]{
      \Jredty{\Gamma{}}{A}{\univ{s}}{\Type{}}
    }{\LRarg{\LRRedTyUniv{s}} : \LRty{\Gamma{}}{A}{\Type{}}}
    \and
    \inferrule[Neutral-Red-Type]{
      \Jredty{\Gamma{}}{A}{T}{\Type{}}\\
      \neutral{T}\\
      \Jconvty{\Gamma}{T}{T}{\univ{s}}{\Type{}}
    }{\LRarg{\LRRedTyNe{}} : \LRty{\Gamma{}}{A}{s}}
    \and
    \inferrule[$\P{}$-Red-Type]{
      \Jredty{\Gamma{}}{A}{\P{}(\ctxvar{x}{X}{s_1})\,Y}{s_2}\\
      \Jty{\Gamma{}}{X}{s_1}\\
      \Jty{\Gamma{}, \ctxvar{x}{X}{s_1}}{Y}{s_2}\\
      \Jconvty{\Gamma{}}{\P{}(\ctxvar{x}{X}{s_1})\,Y}{\P{}(\ctxvar{x}{X}{s_1})\,Y}{s_2}\\
      \LRctx{\Gamma{}}{\ctxvar{x}{X}{s_1},\ctxvar{y}{Y}{s_2}}
      }{\LRarg\LRRedTyPi{}:\LRty{\Gamma{}}{A}{s}}
    \and
    \inferrule[Red-Type]{
      \Jredty{\Gamma{}}{A}{K(\overline{t})}{s}\\
      K \in \Sigma\\
      \Jconvsub{\Gamma{}}{\overline{t}}{\overline{t}}{\params{K}}\\
      \LRsub{\Gamma{}}{\overline{t}}{\params{K}}{[\params{K}]}
    }{\LRarg{\LRRedTyK}:\LRty{\Gamma{}}{A}{s}}
  \end{mathpar}
  \end{small}
  \caption{Reducibility for types}
  \label{fig:red-types}
\end{figure}
\begin{figure}
\begin{small}
  \[
    \begin{array}{lcl}
      \LRtm{\Gamma{}}{t}{A}{\Type{}}{\LRRedTyUniv{s}}
      &:=
      & \LRty{\Gamma}{t}{s}\qquad \text{\color{gray} (universe level decreases)}\\
      \LRtm{\Gamma{}}{t}{A}{s}{\LRRedTyNe{}}
      &:=
      & \LRneutraltm{\Gamma}{t}{A}{s} \\
      \LRtm{\Gamma{}}{t}{A}{s}{\LRRedTyPi~\ldots}
      &:=
      & \text{(omitted)} \\
      \LRtm{\Gamma{}}{t}{A}{s}{\LRRedTyK{}~K~\overline{a}~\ldots}
      &:=
      & \LRredK{\Gamma}{t}{K(\overline{a})}{s}
    \end{array}
  \]
  \end{small}
  \caption{Reducibility for terms}
  \label{fig:red-terms}
\end{figure}

\begin{figure}
\begin{small}
  \begin{mathpar}
    \inferrule[Constant-Reducible]{
      \Jredtm{\Gamma{}}{t}{w}{K(\overline{a})}{s}
      \and \eqRel{\Gamma{}}{w}{w}{K(\overline{a})}{s}\\
       \Kprop{\Gamma{}}{w}{K(\overline{a})}{s}}
    {\LRredK{\Gamma{}}{t}{K(\overline{a})}{s}} 
    \ilabel{infrule:constant-reducible}
    \and
    \inferrule[Neutral-Reducible]{\LRneutraltm{\Gamma{}}{n}{K(\overline{a})}{s} }{\Kprop{\Gamma{}}{n}{K(\overline{a})}{s}}
    \ilabel{infrule:neutral-reducible}
    \and
    \inferrule[Inert-Constant-Positive-Reducible]{
      (K, \Inert{K}) \in \Sigma{}\\
      c \in \Inert{K} \\
      \LRsub{\Gamma{}}{\overline{p}}{\params{c}}{[\params{c}]} \\
      \forall i, \LRarg{[v_i]} : \LRredK{\Gamma{}}{v_i}{\dom{c}_i[\overline{p}]}{s}\\
      \cod{c} = K(\overline{u}_c)
      \and
      \eqRel{\Gamma{}}{\overline{u}_c[\overline{p}]}{\overline{a}}{\params{K}}{}}
    {\Kprop{\Gamma{}}{c(\overline{p},\overline{v})}{K(\overline{a})}{s}}
    \ilabel{infrule:inert-constant-positive-reducible}
    \and
    \inferrule[Inert-Constant-Negative-Reducible]{
      (K, \Active{K}) \in \Sigma{},\enspace \Active{K} = \overline{d}\\
      \forall i, {\begin{dcases*}
          \LRarg{[d]_i} :\Kprop{\Gamma{}}{d_i(\overline{a}, t)}{K(\overline{t}[\overline{a},\overline{d(\overline{a},t)}_{<i}])}{s}
          & if $\cod{d_i} = K(\overline{t})$\\
          \LRarg{[d]_i} :\LRtm{\Gamma{}}{d_i(\overline{a},t)}{\cod{d_i}[\overline{a},\overline{d(\overline{a},t)}_{<i}]}{\codsort{d_i}}{[\cod{d_i}]~[\overline{a}]~\overline{[d]}_{<i}}&
          otherwise
        \end{dcases*}}
    }
    {\Kprop{\Gamma{}}{t}{K(\overline{a})}{s}}
    \ilabel{infrule:inert-constant-negative-reducible}
  \end{mathpar}
  \end{small}
  \caption{Reducibility of terms at constant types}
  \label{fig:red-constant}
\end{figure}

\subsection{Fundamental lemma}

At a high level, the fundamental lemma states that derivable judgments are
valid.
More precisely, it consists of a family of lemmas for each judgments of
\MuTT{}:
\begin{theorem}[Fundamental lemma] Let $\Sigma$ be a well-formed
  signature.
  \label{thm:fundamental-lemma}
  \begin{enumerate}
  \item If $\Jctx{\Gamma}$ then $\LRctx{\Gamma}$;
  \item If $\Jty{\Gamma}{A}{s}$ then there is a proof $\LRarg{[A]} : \LRty{\Gamma}{A}{s}$;
  \item If $\Jtm{\Gamma}{t}{A}{s}$ then $\LRtm{\Gamma}{t}{A}{s}{[A]}$
  \item If $\Jconvty{\Gamma}{A}{B}{s}$ then $\LRconvty{\Gamma}{A}{B}{s}{[A]}$
  \item If $\Jconvtm{\Gamma}{t}{u}{A}{s}$ then $\LRconvtm{\Gamma}{t}{u}{A}{s}{[A]}$
  \end{enumerate}
\end{theorem}

The proof of the fundamental lemma proceed by induction on the typing
derivation, generalizing the result to be proved by uniformly closing
reducibility under substitution and extensionality (see the definition of
$\LRty{\Gamma{}, \LRpair{\Delta}{[\Delta]}}{A}{s}$ in \cref{fig:red-context})
and proving the result mutually for all judgments.
The case of reducibility of universes, constant types, introduction of inert
constants for positive types and introduction of active constants for negative
types are mostly straightforward: we organized the logical relation so that
there is already a case available for these forms.
The challenging and interesting cases are thus the dual ones that are not
explicitly mentioned in the logical relation: the introduction of
active constants for positive types and universes and the introduction of inert
constants for negative types.
We sketch the proof for the first case, highlighting some required properties of
\MuTT{} participating to its  design.
\begin{proof}
Consider a typing derivation ending with the rule \nameref{infrule:active-term},
with conclusion
$\Jtm{\Gamma}{d(\overline{t},u)}{\cod{d}[\overline{t},u]}{\codsort{d}}$ where
$\cod{d} = K(\overline{u_d})$ and $(K, \Inert{K})$ is a positive type according
to $\Sigma$.
By induction hypothesis, we have that $d \in \Sigma$, so that the parameters,
domain and codomain of $d$ are reducible and
\begin{align*}
  \LRarg{[\overline{t}]} &: \LRsub{\Gamma}{\overline{t}}{\params{d}}{[\params{d}]},&
  \LRarg{[u]} &: \LRtm{\Gamma}{u}{\dom{d}[\overline{t}]}{\domsort{d}}{[\dom{d}]~[\overline{t}]}
\end{align*}
Since $u$ is reducible, it reduces to a weak head normal form $w$, reducible at
the same type, that is $\LRarg{[w]}
:\LRredK{\Gamma}{w}{K(\overline{u_d}[\overline{t}])}{\domsort{d}}$.
By anti-reduction, it is enough to show that $\acst{d}{\overline{t}}{w}$ is
reducible, which we do by induction on $[w]$, generalizing over the reducible
parameters $\overline{t}$.
By inversion, $\LRarg{[w]}$ is necessarily produced with an instance 
of \nameref{infrule:constant-reducible}, and is either a neutral or of the shape
$w = c(\overline{p},\overline{v})$ for $c \in \Inert{K}$ (by \nameref{infrule:inert-constant-positive-reducible}).
If $w$ is neutral or $\neg\react{d}{c}$ holds, then $\acst{d}{\overline{t}}{w}$ is
neutral, can be shown to be well-typed using the escape lemma, so it is
reducible.
Otherwise, $\react{d}{c}$ ensures that there exist a rewrite rule
$\rewrule{\Delta}{\sigma}{\Delta_{\lin{}}}{d}{\overline{x}}{pat}{r} \in
\RewRules{}$ such that $\acst{d}{\overline{x}}{pat}$ unifies with
$\acst{d}{\overline{t}}{w}$ thanks to linearity, yielding a reducible
substitution $\LRsub{\Gamma}{\rho}{\Delta_\lin}{[\Delta_{\lin}]}$.
Using the premise of \nameref{infrule:inert-constant-positive-reducible}
obtained from $[w]$, we have that $w = c(\overline{p}, \overline{v})$, $\cod{c}
= K(\overline{u_c})$ and
$\Jconvsub{\Gamma}{\overline{u_c}[\overline{p}]}{\overline{u_d}[\overline{t}]}{\params{K}}$.
By linearizability of the rewrite rule (\cref{def:linearizability}), we obtain a
renaming $\tau$ that is an inverse of $\sigma$ up to conversion.
The composed substitution $\tau[\rho]$ is reducible,
$\LRsub{\Gamma}{\nu[\rho]}{\Delta}{[\Gamma]}{[\Delta]}$,  because $\rho$ is
reducible and $\tau$ is a renaming.
By substitution into the reducibility proof of the right hand side $r$ obtained
from the reducibility of the signature $[\Sigma]$ together with the induction
hypothesis on $\LRarg{[w]}$ for recursive occurrences from the pattern, we have
that $r[\tau[\rho], \occrecsub{d}{pat}[\tau[\rho]]]$ is reducible at type
\[\cod{d}[\overline{x}, \erasepat{pat}][\sigma{}][\tau[\rho]] \enspace\conv\enspace
  \cod{d}[\overline{x},\erasepat{pat}][\rho] \enspace\conv\enspace \cod{d}[\overline{t},u].\]
Finally, by anti-reduction $\acst{d}{\overline{t}}{w}$ is reducible at the adequate type.
The case of an eliminator over a universe follows the same pattern, the
main modification being the organisation of the inductive hypothesis coming from
the signature.
For the case of a builder of negative type, the general case builds a
reducibility proof now by coinduction.
\end{proof}

\subsection{Consequences}

Using the fundamental lemma and the definition of the logical relation on
positive inert types, we obtain as a direct consequence that any term
$\Jtm{\emptyContext}{t}{K(\overline{u})}{\Type{}}$ of positive type is convertible to
an inert constant introducing $K$ or a neutral term.
The following lemma ensures that there is no neutral term in $\Type$, so
\cref{thm:canonicity} and its immediate corollary \cref{thm:consistency} follow.

\begin{lemma} Closed neutrals belong to isolated sorts:
  \[
    \Jtm{\emptyContext}{n}{A}{s} \wedge \neutral{n} \implies \isolated{s} \vee (A = \univ{s'} \wedge \isolated{s'})
  \] 
\end{lemma}
\begin{proof}
The proof proceed by induction on the neutrality of $n$.
The variable case is impossible since the context is empty, and the application
case $n = n'~t$ proceed by induction on $n'$ using inversions on the typing
derivation to show that $n'$ is typed with a $\Pi$ type in the empty context,
at the same sort $s$ as $n$, hence $s$ is isolated.
The important case $n = \acst{d}{\overline{p}}{w}$ consider two cases depending on the
neutrality of $w$.
If $w$ is not neutral, $d \in \Sigma$ does not satisfy progress, so its
codomain must satisfy the conclusion of the lemma.
If $w$ is neutral, it is again well-typed in an empty context by inversion on
the typing derivation, so the type of $w$ satisfies the conclusion of the lemma by
induction hypothesis, and we conclude because the active constant $d \in \Sigma$
must preserve isolation.
\end{proof}

Decidability of conversion is proven by defining an algorithmic
version of the conversion of two terms $t$ and $u$ which basically
amounts to computing the whnf of $t$ and $u$, compare their head, and
apply the algorithm recursively if necessary.
Correctness of algorithmic conversion is easy as the rules used are
particular cases of typed conversion (\cref{fig:MuTT-conv}).
Then, it is shown that this algorithmic conversion is also complete by
replaying \cref{thm:fundamental-lemma} with a definition of the
logical relation using algorithmic conversion instead of typed
conversion.
Actually, the formalization of \cite{AbelOV18} factorizes the two
proofs of the fundamental lemma by defining an abstract interface to both
algorithmic conversion and typed conversion and use this interface
in  the definition of the logical relation instead.
Then, to get decidability of type checking, we can simply rely on the
work of \citet{LennonBertrand2021} on bidirectional type-checking,
which defines an algorithmic version of type-checking provided 
that the theory enjoys subject reduction and decidability of conversion.  
  
\section{Extensionality}
\label{sec:extensionality}

Our parametrization of \MuTT in \cref{sec:mutt} only
allows us to extend conversion through the introduction of new rewrite
rules.
However, some extensions of conversion such as extensionality principles
are inherently undirected and cannot be specified with reduction rules, 
but must directly extend conversion.
In this section, we do not provide a generic mechanism to enrich conversion with
extensionality principles, a challenging goal that we leave for future work,
but we remark that the logical relation naturally justifies them on two
compelling examples.

\paragraph{Primitive projections}

The definition of negative dependent sum in \cref{sec:ind} (and
more generally, any record type) can be equipped with the following
extensionality principle:
\begin{align*}
  \inferrule{\Jtm{\Gamma}{t}{\SigmaName~A~B}{\Type}}{\Jconvtm{\Gamma}
  {\pairName ( A , B , \fstName (A , B,  t) , \sndName ( A , B ,  t))} {t}{\SigmaName~A~B}{\Type} }
\end{align*}
This conversion rule can be added to the system by postulating it when
$t$ is neutral.
The logical relation framework then straightforwardly shows that the
conversion is valid on any term, as any term of type $ \SigmaName~A~B$
reduces to a whnf which is either a $\pairName$, in which case the equality
holds by computation of  $\fstName$ and $\sndName$; or it is a neutral
term, in which case the equality holds with the new conversion rule.
Then, it suffices to remark that algorithmic conversion can also be
extended with this new conversion rule on neutral terms without
compromising decidability. 

\paragraph{Strict Propositions}

\citet{GilbertCST19} propose the introduction of a new sort \SProp{} of
{\em strict propositions} to \Agda\footnote{\url{https://agda.readthedocs.io/en/v2.6.0/language/prop.html}}
and \Coq\footnote{Since \Coq 8.10: \url{https://coq.inria.fr/doc/addendum/sprop.html}}.
The characteristic feature of \SProp{} is its definitional proof irrelevance,
e.g. any two inhabitants $p,q$ of a type $P$ in \SProp{} are convertible:
\begin{align*}
  \label{eq:sprop-definitional-pi}
  \inferrule{\Jtm{\Gamma}{p}{P}{\SProp}\\\Jtm{\Gamma}{q}{P}{\SProp}}{\Jconvtm{\Gamma}{p}{q}{P}{\SProp}} 
\end{align*}
%

%
As explained by~\citet{GilbertCST19}, to encode such a sort of strict
proposition in \MuTT{}, it is enough to introduce a new sort \SProp{} with a
single empty inductive type $\Jty{}{\bot}{\SProp}$ and its eliminator to
$\Type{}$, together with a conversion rule equating any two neutral terms at
this type.

\section{Related work}
\label{sec:related-work}

\paragraph{Extending type theories}
Pure Type Systems (PTS)~\cite{books/daglib/0032840} is a general framework for defining type
theories based on $\lambda$-calculus extended with additional sort constants.
Metatheoretical results such as consistency, subject reduction and normalisation
have been established for classes of PTS and their extensions, for instance with
cumulativity~ \cite{Luo90} but the computational content is usually entirely
defined from $\beta$-reduction.
\citet{AllaisMB13} extend conversion with a fixed set of additional equations
between neutral terms for a simply typed language, with type theory left as a
future work goal.
Their work use a similar methodology with a logical relation to show that
conversion can be reduced to a standard reduction path.
\citet{BasoldG16} propose a type theory with a generic scheme to define inductive
and coinductive types uniformly.
\MuTT{} develops beyond their treatment in two orthogonal directions, supporting
universes and allowing types that are not necessarily inductive or coinductive.

\paragraph{Rewriting in type theory}

Our setting to define rewrite rules is based on the recent work of
\citet{CockxTW21} but combining rewrite systems and type systems stems
from the work of \citet{Breazu-Tannen88a}, extending simply typed
lambda-calculus with higher-order rewrite rules.
This framework was later taken to dependent type theory by \citet{BarbaneraFG97}.
They extend the Calculus of Constructions with first- and
higher-order rewrite rules, provided the higher-order rules do not
introduce any critical pairs.
%
\citet{Walukiewicz-Chrzaszcz03} prove subject reduction
for another variant of the Calculus of Constructions with a more
general notion of higher-order rewrite rules
and 
completeness and consistency of this
system has been studied in~\citep{Walukiewicz-ChrzaszczC06}.
The Calculus of Algebraic Constructions~\citep{Blanqui05} is another
extension of the Calculus of Constructions with a restricted form of
higher-order rewrite rules. It also provides criteria for checking
subject reduction and strong normalization.
All those work serves as the base to our present work, and we do not
claim any originality with respect to our termination criteria which
is basically enforced by typing conditions in the definition of a well-formed signature (\cref{sec:well-form-sign}).

\paragraph{Modal type theories}
Modalities have recently gained traction to extend type theory in a variety of
directions~\cite{SchreiberS14,RijkeSS19,Shulman18,Kavvos19,NuytsD18,BirkedalCMMPS20}, supporting
the addition of new logical principles and constraints on the structure of type
theoretical judgments.
In order to accommodate the zoo of modalities required for different
applications, general frameworks parametrized by a $2$-category of modes have
been proposed, first in a simply typed setting~\cite{LicataSR17}, and gradually
being adapted to a dependent setting
~\cite{BirkedalCMMPS20,GratzerKNB20,GratzerSB19}.
Modalities, in particular non-lex comonads that modify the action of
substitutions, encompass a wider setting than what we present in this work at
the cost of a more complex metatheory.
To this day, beyond state-of-the-art experiments, no proof assistant
implementation support parametrized modalities.
Our approach assumes standard context management and substitution propagation so
should be more readily compatible with existing mainstream proof assistant such
as Coq or Agda for an implementation.

\paragraph{Logical relations, type theory and categorical models}
Since Plotkin's seminal work~\cite{Plotkin73}, logical relations have been used
pervasively to prove metatheoretical properties of programming languages and
type theory~\cite{Mitchell91}.
A categorical perspective on these techniques have been developed over the last
three
decades~\cite{MitchellS92,Fiore02,Shulman15,SterlingSpitters18,SterlingHarper20},
providing efficient but rarely effective methods to prove normalization.
\citet{AbelAD07} apply these techniques to dependent type type theory, while
\citet{Coquand19,Coquand21} uses a so-called reduction-free variant of logical
relations.
\citet{AbelOV18} provide the first mechanization of logical relations to prove
decidability of type checking of type theory in itself, on which we build.
Such mechanized developments remain to date a difficult task as
witnessed by the recent POPLMark reloaded challenge~\cite{AbelAHPMSS19}.

\section{Conclusion}
\label{sec:conclusion}

We have presented a generic multiverse type theory \MuTT in which multiple,
possibly incompatible type universes can safely cohabit without endangering its
meta-theoretical properties.
This new sort system provides a type-theoretic mechanism to separate
incompatible computational or logical features, which can further be used to
mediate between universes, \eg using specific new constants that make bridges
between universes.
Beyond the simple instances that we present here, we expect that many models of
\MLTT{}~\cite{BoulierPT17,pedrotTabareau:popl2020,AltenkirchBKT19} have interesting presentations in \MuTT{} capturing their
computational behaviour.
Parametrized extensions of conversion as presented in~\cref{sec:extensionality}
is an important future milestone to that endeavour.
A natural next step is to make the theory sort-polymorphic, so that
sort-agnostic definitions can be shared more easily between universes, extending
the existing universe-level polymorphism that is implemented in today's proof
assistants \cite{AgdaUniverses,univ-poly}.


\bibliography{refs,strings,pleiad,bib,common}



\end{document}